\documentclass[11pt]{amsart}
\usepackage{graphicx,amssymb,amsmath,amsthm}
\usepackage{enumerate}
\usepackage{comment,cite,color}
\usepackage{algorithmic}
\usepackage{mathrsfs}
\usepackage[ruled]{algorithm}
\usepackage{epsfig}

\numberwithin{equation}{section}

\newcommand{\PP}{\mathbb{P}}

\newcommand{\abs}[1]{\lvert#1\rvert}

\renewcommand{\le}{\leq}
\renewcommand{\ge}{\geq}
\newcommand{\tp}{\mathsf{T}}
\newcommand{\cN}{\mathcal{N}}
\newcommand{\lcst}{\theta} %replace \theta_{-}
\newcommand{\ucst}{\omega}  %replace \theta_{+}
\newcommand{\Ndim}{n} %replace N by n
\newcommand{\I}{\Omega}

\newcommand{\E}{\mathbb E}

\newcommand{\R}{{\mathbb R}}

\newcommand{\N}{{\mathbb N}}

\renewcommand{\eqref}[1]{(\ref{#1})}

\newtheorem{prop}{Proposition}[section]
\newtheorem{lem}[prop]{Lemma}

\newtheorem{coro}[prop]{Corollary}
\newtheorem{theo}[prop]{Theorem}

\usepackage[top=0.6in,bottom=0.6in,left=0.7in,right=0.7in]{geometry}
\date{}

\begin{document}
\bibliographystyle{plain}

\title{Robustness Properties of Dimensionality Reduction with  Gaussian Random Matrices }

\thanks{Research of Bin Han was supported by
the Natural Sciences and Engineering Research Council of Canada (NSERC Canada Grant No. 05865).
Research of Zhiqiang Xu was supported  by NSFC grant (11171336, 11422113,11021101, 11331012) and by National Basic Research Program of China (973 Program 2015CB856000).
}

\author{Bin Han}
\address{Department of Mathematical and Statistical Sciences,
University of Alberta, Edmonton, Alberta T6G 2G1, Canada,\qquad bhan@ualberta.ca}

\author{Zhiqiang Xu}
\address{LSEC, Inst.~Comp.~Math., Academy of
Mathematics and System Science,  Chinese Academy of Sciences, Beijing 100091, China, \qquad xuzq@lsec.cc.ac.cn}


 \maketitle

\maketitle

\maketitle

\begin{abstract}
In this paper we study the robustness properties of dimensionality reduction with Gaussian random matrices having arbitrarily erased rows. We first study the robustness property against erasure for the almost norm preservation property of Gaussian random matrices
by obtaining the optimal estimate of the erasure ratio for a small given norm distortion rate. As a consequence, we establish the robustness property of Johnson-Lindenstrauss lemma and the robustness property of restricted isometry property with corruption for Gaussian random matrices.
Secondly, we obtain a sharp estimate for the optimal lower and upper bounds of norm distortion rates of Gaussian random matrices under a given erasure ratio. This allows us to establish the strong restricted isometry property with the almost optimal RIP constants,
which plays a central role in the study of phaseless compressed sensing.
\end{abstract}

\section{Introduction and Motivations}

In this paper we are interested in investigating various robustness properties
of dimensionality reduction with
Gaussian random matrices having arbitrarily erased rows. Then we shall use the results to study the robustness properties of the Johnson--Lindenstrauss lemma and restricted isometry property.

Throughout the paper, $A=(a_{j,k})_{1\le j\le m, 1\le k\le \Ndim}\in \R^{m\times \Ndim}$ will be a Gaussian random matrix such that each entry $a_{j,k}$ is an independent identically distributed (i.i.d.) random variable under the standard Gaussian/normal distribution $\cN(0,1)$ with zero mean and unit standard deviation. For $T\subseteq \{1,\ldots,m\}$, we shall adopt the notation $A_T \in \R^{\abs{T}\times \Ndim}$ to denote
the $\abs{T}\times \Ndim$ sub-matrix of $A$ by keeping the {\em rows} of $A$ with the row indices from $T$, where $|T|$ is the cardinality of the set $T$.
Let $x_0\in \R^{\Ndim}$ be a fixed vector with $\|x_0\|=1$, where $\|x_0\|$ is the Euclidean norm of the vector $x_0$. For $\epsilon>0$ and $0\le \beta<1$, we define
\begin{equation}
\I_{\epsilon,\beta}:=\I_{\epsilon,\beta}(A,x_0):=\left\{ \left|\tfrac{1}{\abs{T}}\|A_T  x_0\|^2-1\right| \le \epsilon\; \text{for all} \; T\subseteq \{1,\ldots,m\}\; \mbox{satisfying}\; |T^c|\le \beta m \right\},\label{I:epsbeta}
\end{equation}
where $T^c:=\{1,\ldots,m\}\backslash T$. For every fixed $\epsilon>0$, it follows from the definition in \eqref{I:epsbeta} that $\PP(\I_{\epsilon,\beta})$ is a decreasing function in terms of $\beta$, where the probability is taken over the Gaussian random matrix $A$.

It is well known in the literature by standard tail-bounds for the chi-squared distribution
(e.g., see \cite[Lemma~4.1]{A:jcss:2003})
that
\begin{equation}\label{eq:epsilon:0}
\begin{split}
&\PP\{\tfrac{1}{m}\|Ax_0\|^2>1+\epsilon\}<e^{-(\epsilon^2/4-\epsilon^3/6)m},\qquad \\ &\PP\{\tfrac{1}{m}\|Ax_0\|^2<1-\epsilon\}<e^{-(\epsilon^2/4-\epsilon^3/6)m},
\end{split}
\qquad \forall\, m\in \N,  0<\epsilon<1.
\end{equation}
Consequently, with high probability, a normalized Gaussian random matrix $\frac{1}{\sqrt{m}} A$ has the following almost norm preservation property:
\begin{equation}\label{grm:normpreservation}
\PP(\I_{\epsilon,0})=\PP\left \{\left|\tfrac{1}{m}\|Ax_0\|^2-1\right|\le \epsilon\right\}
\ge 1-2e^{-(\epsilon^2/4-\epsilon^3/6)m},\qquad \forall\, m\in \N, 0<\epsilon<1.
\end{equation}
The inequality in \eqref{grm:normpreservation} also implies the Johnson-Lindenstrauss lemma (see \cite{JL,A:jcss:2003}).
For $N$ points $p_1,\ldots, p_N\in \R^{\Ndim}$ and for $0<\epsilon<1$,
the Johnson-Lindenstrauss lemma says that for $m=O(\tfrac{\ln N}{\epsilon^2 })$, there exists a projection matrix $A\in \R^{m\times \Ndim}$ such that the following almost norm preservation property holds:
\begin{equation}\label{JLlemma}
(1-\epsilon)\|p_j-p_k\|^2 \le \| Ap_j-A p_k\|^2\le (1+\epsilon)\|p_j-p_k\|^2, \qquad \forall\; 1\le j,k\le N.
\end{equation}
To establish the above almost norm preservation property in \eqref{JLlemma}, the projection matrix $A$ is often taken to be a random matrix so that the almost norm preservation property in \eqref{grm:normpreservation} holds with high probability. The Johnson-Lindenstrauss lemma is a fundamental technique to reduce the dimensionality of the data and has many applications in information theory, machine learning and algorithms (c.f. \cite{JLSur} and references therein).

In the compressed sensing literature, the restricted isometry property (RIP) matrix is also related to \eqref{grm:normpreservation}.
For $x\in \R^\Ndim$ and $s\in \N$, we say that $x$ is \emph{$s$-sparse} if $x$ has no more than $s$ nonzero entries. Under a measurement matrix $A\in \R^{m\times \Ndim}$,
we have  $y:=(y_1,\ldots, y_m)^\tp:=Ax$ with $m$ measurements $y_1,\ldots, y_m$. To successfully recover the unknown sparse signal $x$ from the measurement vector $y$, it is important for the matrix $A$ to satisfy the \emph{restricted isometry property} (RIP) with a small positive RIP constant $0<\epsilon_s<1$:
\begin{equation}\label{rip:const}
(1-\epsilon_s)\|v\|^2 \le \|A v\|^2\le (1+\epsilon_s)\|v\|^2, \qquad \mbox{for all $s$-sparse vectors}\; v\in \R^{\Ndim}.
\end{equation}
The above restricted isometry property with a small positive RIP constant $\epsilon_s$
is often established by considering $A$ to be a random matrix such as a normalized Gaussian random matrix so that the almost norm preservation property in \eqref{grm:normpreservation} holds with high probability for $\epsilon=\epsilon_s$ (see \cite{BDDW:ca:2008,KW:SMA:2011}).

When $\beta>0$, we suppose that at most $\beta m$ rows of the Gaussian random matrix $A$ are arbitrarily erased.
It is of interest in both theory and application to study how large is the erasure ratio $\beta$ so that a normalized Gaussian random matrix $\frac{1}{\sqrt{m}} A$ with any arbitrarily erased $\beta m$ rows still has the almost norm preservation property with high probability.
Particularly, we are interested in  the following two problems:

\begin{enumerate}[{\bf Problem } 1 :]
\item Give $0<\epsilon<1$ and $0<\alpha<1$, what is the maximum $\beta$ so that
\[
\PP(\I_{\epsilon,\beta}) \ge 1-3 e^{-\alpha (\epsilon^2/4-\epsilon^3/6) m}, \quad \text{for all}\quad m\in \N.
\]
\item Given $0<\beta < 1$ and $\alpha>0$, what is the minimum $\epsilon$ so that
\[
\PP(\I_{\epsilon,\beta}) \geq 1-2\exp(-\alpha m), \quad \text{for all}\quad m\in \N.
\]
%Here $c$ is a positive constant.
\end{enumerate}

Let us briefly explain our motivation for considering $\PP(\I_{\epsilon,\beta})$ with $\beta>0$ in the setting of Johnson-Lindenstrauss lemma and of compressed sensing.
In Johnson-Lindenstrauss lemma, note that each projected vector $Ap_j$ has $m$ entries.
The projected vectors are often transmitted through network by $m$ parallel channels, that is, each entry of $Ap_j$ is transmitted through an independent channel in a parallel manner. If some channels are out of work, we can only receive the corrupted projected vectors $A_T p_j$ instead of $A p_j$ for $j=1,\ldots, N$, where $T\subseteq \{1,\ldots, m\}$ is an unknown subset  with $|T^c|\le \beta m$ for some given corruption/erasure ratio $0<\beta<1$. Consequently,
it is important to first establish the
almost norm preservation property with high probability in \eqref{grm:normpreservation} for $\I_{\epsilon,\beta}$ with $\beta>0$.
The compressed sensing with corruption considers the problem that a certain portion of the obtained measurements $y_1,\ldots, y_m$ are  missing or corrupted by sparse noises (e.g., see \cite{Li,SHB,WM} and references therein). In other words, one can only obtain the measurements $A_T x$ for some unknown subset $T\subseteq \{1,\ldots, m\}$ such that $|T^c|\le \beta m$ for some given corruption/erasure ratio $0<\beta<1$.
Therefore, it is important that  the matrices $A_T$ have the restricted isometry property with a small positive RIP constant $\epsilon_s$ for all subsets $T\subseteq \{1,\ldots, m\}$ with $|T^c|\le \beta m$. Particularly, in compressive phase retrieval, to recover
sparse signals from the magnitude of the linear measurement, one introduces the concept  of {\em strong RIP} which requires  the matrices  $A_T$ satisfy the RIP property for all
subsets  $T\subseteq \{1,\ldots, m\}$ with $|T^c|\le \beta m$ (c.f. \cite{VX}). For example, in \cite{VX}, the authors considered  the case $\beta=1/2$. To achieve this robustness property, it is natural to first establish the
almost norm preservation property with high probability by replacing $\I_{\epsilon,0}$ and $\epsilon$ in \eqref{grm:normpreservation} with $\I_{\epsilon_s,\beta}$ and $\epsilon_s$, respectively for $\beta>0$.

We first consider Problem 1.
To study how large is the erasure ratio $\beta$ so that a normalized Gaussian random matrix $A$ with arbitrarily erased $\beta m$ rows still has the almost norm preservation property with high probability, we introduce a quantity $\beta^{\max}_{\epsilon,\alpha}$ to characterize the largest possible such erasure ratio $\beta$ with a given fixed high probability rate $\alpha>0$. For $\epsilon>0$ and $\alpha>0$, we define
\begin{equation}\label{betamax}
\beta^{\max}_{\epsilon,\alpha}:=\sup\{0\le \beta<1\; :\; \PP(\I_{\epsilon,\beta}) \ge 1-3 e^{-\alpha (\epsilon^2/4-\epsilon^3/6) m}\; \text{for all}\; m\in \N\}.
\end{equation}
If the above set in the right-hand side of \eqref{betamax} is empty, then we simply define $\beta^{\max}_{\epsilon,\alpha}:=0$. Due to \eqref{grm:normpreservation} and $\I_{\epsilon,\beta}\subseteq \I_{\epsilon,0}$ for all $0\le \beta<1$, it makes sense for us to only consider $0<\alpha\le 1$.
The multiplicative constant $3$ before $e^{-\alpha (\epsilon^2/4-\epsilon^3/6) m}$ in \eqref{betamax} is not essential and can be replaced by any absolute constant greater than $2$. For simplicity of presentation, we stick to the constant $3$ in \eqref{betamax}.

For $\epsilon>0$ and $0\le \beta<1$, a closely related notion to $\I_{\epsilon,\beta}$ in \eqref{I:epsbeta} is
\begin{equation}\label{mI:epsbeta}
\mathring{\I}_{\epsilon,\beta}:=\mathring{\I}_{\epsilon,\beta}(A,x_0):=\left\{ \left|\tfrac{1}{m}\|A_T  x_0\|^2-1\right| \le \epsilon\; \text{for all} \; T\subseteq \{1,\ldots,m\}\; \mbox{satisfying}\; |T^c|\le \beta m \right\}.
\end{equation}
That is, we used the uniform normalization factor $\frac{1}{m}$ for $\mathring{\I}_{\epsilon,\beta}$ in \eqref{mI:epsbeta} instead of the factor $\frac{1}{|T|}$ for $\I_{\epsilon,\beta}$ in \eqref{I:epsbeta}.
Similar to \eqref{betamax}, for $\epsilon>0$ and $\alpha>0$, we define
\begin{equation}\label{mbetamax}
\mathring{\beta}^{\max}_{\epsilon,\alpha}:=\sup\{0\le \beta<1\; :\; \PP(\mathring{\I}_{\epsilon,\beta}) \ge 1-3 e^{-\alpha (\epsilon^2/4-\epsilon^3/6)m}\; \text{for all}\; m\in \N\}.
\end{equation}

For the case $\epsilon\to 0^+$ (that is, $\epsilon$ is small for the almost norm preservation property in \eqref{grm:normpreservation}), we have the following result.

\begin{theo}\label{thm:optimal}
Let $A$ be an $m\times \Ndim$ random matrix with independent identically distributed entries obeying $\cN(0,1)$.
For every $0<\alpha<1$,
\begin{equation} \label{eq:optimal}
\left(\frac{1-\sqrt{\alpha}}{32}\right) \frac{\epsilon}{\ln \tfrac{1}{\epsilon}} < \beta^{\max}_{\epsilon,\alpha}< \left(\frac{2+2\epsilon_g}{c_g^2\epsilon_g}\right) \frac{\epsilon}{\ln \tfrac{1}{\epsilon}}, \qquad \, 0<\epsilon<\min(\tfrac{1-\sqrt{\alpha}}{4},\epsilon_g,4\epsilon_g^2)
\end{equation}
and
\begin{equation} \label{eq:optimal:0}
\left(\frac{1-\sqrt{\alpha}}{32}\right) \frac{\epsilon}{\ln \tfrac{1}{\epsilon}} < \mathring{\beta}^{\max}_{\epsilon,\alpha} < \left(\frac{1}{4c_g^2}\right) \frac{\epsilon}{\ln \tfrac{1}{\epsilon}}, \qquad \, 0<\epsilon<\min(\tfrac{1-\sqrt{\alpha}}{4},c_g^2\ln2, \tfrac{1}{2c_g^2}),
\end{equation}
where $c_g$ and $\epsilon_g$ are absolute positive constants given in \eqref{cg} and \eqref{epsg}, respectively.
\end{theo}

Theorem \ref{thm:optimal} shows that  $\beta^{\max}_{\epsilon,\alpha}=O(\frac{\epsilon}{\ln 1/\epsilon})$ has the optimal order when $\epsilon $ is small enough. Hence, Theorem~\ref{thm:optimal} presents a solution to Problem 1 up to a multiplicative constant provided that $\epsilon$ is small enough.
As a direct consequence of Theorem~\ref{thm:optimal} (more precisely, Theorem~\ref{thm:lowerbound:0}),
by the standard argument in the literature for proving the Johnson-Lindenstrauss lemma using random matrices, we have the following robust version of the Johnson-Lindenstrauss lemma.

\begin{coro}\label{cor:JL}
Let $0<\alpha<1$ and $0<\epsilon<\frac{1-\sqrt{\alpha}}{4}$.
Let $N,m,n\in \N$ such that $m>\frac{\ln(3N(N-1)/2)}{\alpha(\epsilon^2/4-\epsilon^3/6)}$.
Let $A$ be an $m\times \Ndim$ random matrix with i.i.d. entries obeying $\cN(0,1)$.
For any given $N$ points $p_1,\ldots, p_N\in \R^n$, with probability at least $1-\frac{3}{2}N(N-1)e^{-\alpha (\epsilon^2/4-\epsilon^3/6)m}>0$,
\begin{equation}\label{JLlemma:robust}
(1-\epsilon)\|p_j-p_k\|^2 \le \tfrac{1}{m}\| A_Tp_j-A_T p_k\|^2\le (1+\epsilon)\|p_j-p_k\|^2, \quad \forall\; 1\le j,k\le N\;\; \mbox{and}\;\; T\in T_{\epsilon,\alpha},
\end{equation}
where $T_{\epsilon,\alpha}$ is defined to be
\begin{equation}\label{T:erasureration}
T_{\epsilon,\alpha}:=\left \{T\subseteq \{1,\ldots, m\} \; : \;
|T^c|\le m \left(\frac{1-\sqrt{\alpha}}{32}\right)\frac{\epsilon}{\ln \tfrac{1}{\epsilon}}\right\}.
\end{equation}
\end{coro}

%By the same simple argument using union bounds as in \cite[Theorem~5.2]{BDDW:ca:2008},
Another consequence of Theorem~\ref{thm:lowerbound:0} is the following result on the robust restricted isometry property.

\begin{coro}\label{cor:rip}
Let $0<\alpha<1$ and $0<\epsilon<\frac{1-\sqrt{\alpha}}{4}$.
Let $s, m,\Ndim\in \N$ satisfy
%$s\leq \frac{\alpha \epsilon^2}{8}\frac{m}{\ln \frac{12 e n}{\epsilon s}}$.
$s \ln \frac{24 en}{\epsilon s}< \alpha(\epsilon^2/16-\epsilon^3/24)m-\ln3$.
Let $A$ be an $m\times \Ndim$ random matrix with i.i.d. entries obeying $\cN(0,1)$. With probability at least $1-3 (\frac{24en}{\epsilon s})^s e^{-\alpha(\epsilon^2/16-\epsilon^3/24)m}>0$,
\begin{equation}\label{rip:robust}
(1-\epsilon)\|v\|^2 \le \tfrac{1}{m}  \|A_T v\|^2\le (1+\epsilon)\|v\|^2, \quad \forall\; \mbox{$s$-sparse $v\in \R^{\Ndim}$ and}\;\; T\in T_{\epsilon/2,\alpha},
\end{equation}
where $T_{\epsilon/2,\alpha}$ is defined in \eqref{T:erasureration}.
\end{coro}

We now turn to Problem 2, which is also related to erasure robust frames (see \cite{W}).
For a given $0<\beta<1$, we would like to determine the minimum $\epsilon $ so that
$\tfrac{1}{\abs{T}}\|A_T  x_0\|^2\in [1-\epsilon, 1+\epsilon]$ with high probability for all $T\subseteq \{1,\ldots,d\}$ satisfying $|T^c|\le \beta m$.
For this purpose, we consider the most general case instead of the particular subsets $\I_{\epsilon,\beta}$ in \eqref{I:epsbeta}. Recall that $x_0\in \R^\Ndim$ with $\|x_0\|=1$.
For $0\le \beta<1$ and $0\le \lcst\le \ucst\le \infty$, we define
\begin{equation}\label{I:interval}
\I_{[\lcst,\ucst],\beta}:=\I_{[\lcst,\ucst],\beta}(A,x_0):=\big\{ \tfrac{1}{\abs{T}}\|A_T  x_0\|^2\in [\lcst,\ucst] \; \forall\, T\subseteq \{1,\ldots,m\} \; \mbox{satisfying}\; \abs{T^c}\leq \beta m \big\}.
\end{equation}
Obviously, $\I_{\epsilon,\beta}$ in \eqref{I:epsbeta} simply becomes $\I_{\epsilon,\beta}=\I_{[1-\epsilon,1+\epsilon],\beta}$.
For $0<\beta<1$ and $\alpha>0$, we define
\begin{align}
&\lcst^{\max}_{\beta}(\alpha):=\sup\{0\le \lcst\le \infty \; : \;  \PP(\I_{[\lcst,\infty],\beta}) \ge 1-\exp(-\alpha m)\;\; \mbox{for all}\; m\in \N\},\label{lcst:alpha}\\
&\ucst^{\min}_{\beta}(\alpha):=\inf\{0\le \ucst\le \infty \; : \; \PP(\I_{[0,\ucst],\beta}) \ge 1-\exp(-\alpha m)\;\; \mbox{for all}\; m\in \N\}, \label{ucst:alpha}
\end{align}
and
\begin{equation} \label{cst}
\lcst^{\max}_\beta:=\sup\{\lcst^{\max}_{\beta}(\alpha) \; : \; \alpha>0\} \quad \mbox{and}\quad
\ucst^{\min}_\beta:=\inf\{\ucst^{\min}_{\beta}(\alpha) \; : \; \alpha>0\}.
\end{equation}
A simple observation from the above definitions is that $\lcst^{\max}_{\beta}(\alpha)\le \lcst^{\max}_{\beta}\le
\ucst^{\min}_{\beta}\le \ucst^{\min}_{\beta}(\alpha)$ and
\begin{equation}\label{P:lcst:ucst}
\PP(\I_{[\lcst^{\max}_{\beta}(\alpha),\ucst^{\min}_{\beta}(\alpha)],
\beta})\geq 1-2\exp(-\alpha m), \qquad \forall\; m\in \N.
\end{equation}
If $0<\lcst^{\max}_{\beta}(\alpha)\le \ucst^{\min}_{\beta}(\alpha)<2$, then Problem~2 is solved with $\epsilon=\max(1-\lcst^{\max}_{\beta}(\alpha),\ucst^{\min}_\beta(\alpha)-1)>0$.
Similar to \eqref{I:interval}, we define
\begin{equation}\label{I:interval:2}
\mathring{\I}_{[\lcst,\ucst],\beta}:=\left\{ \tfrac{1}{m}\|A_T  x_0\|^2\in [\lcst,\ucst]\quad \forall\, T\subseteq \{1,\ldots,m\} \; \mbox{satisfying}\; \abs{T^c}\leq \beta m \right\}
\end{equation}
and we can define $ \mathring{\lcst}^{\max}_{\beta}, \mathring{\ucst}^{\min}_{\beta}$ similar to $\lcst^{\max}_{\beta}, \ucst^{\min}_{\beta}$, respectively by replacing $\I$ with $\mathring{\I}$.

We now briefly explain why we are interested in $\I_{[\lcst,\ucst],\beta}$. An $m\times n$ matrix $A$
is said to have the \emph{strong restricted isometry property} of sparse order $s\in \N$ and level $[\lcst, \ucst, \beta]$ with $0<\lcst\le \ucst<2,
0\leq\beta<1$ if
\begin{equation}\label{rip:strong}
\lcst \|v\|^2 \le \tfrac{1}{m}\|A_T v\|^2\le \ucst \|v\|^2, \quad \forall\, \mbox{$s$-sparse $v\in \R^{\Ndim}$ and}\; T\subseteq \{1,\ldots, m\} \;\mbox{with}\; |T^c|\le \beta m.
\end{equation}
The strong restricted isometry property plays a critical role in the study of phaseless compressed sensing in \cite{BM13, VX}.
In \cite{VX}, the authors investigated the case where $\beta=1/2$ with showing that the Gaussian matrix has the strong RIP of order $s$ and level $[\lcst_0,\ucst_0,1/2]$ with high probability provided $m=O(s\log e n)$. Here $\lcst_0$ and $\ucst_0$ are absolute constants. The original motivation for this work
is to extend the result in \cite{VX} to the arbitrary  $\beta\in [0,1)$.
 To show that there indeed exists a measurement matrix $A$ having the strong restricted property of sparse order $s$ and level $[\lcst,\ucst,\beta]$ in \eqref{rip:strong}, the matrix $A$ is often constructed by an $m\times n$ Gaussian random matrix with i.i.d. entries obeying $\mathcal{N}(0,1)$ and
one would like to have $\PP(\I_{[\lcst,\ucst],\beta})>0$ for $0<\lcst\le \ucst<2$ with the largest possible $\lcst$ and the smallest possible $\ucst$.
That is, if we can prove the inequalities $0<\mathring{\lcst}^{\max}_{\beta}\le  \mathring{\ucst}^{\min}_{\beta}<2$, for any $\lcst, \ucst$ satisfying
$0<\lcst<\mathring{\lcst}^{\max}_{\beta}\le  \mathring{\ucst}^{\min}_{\beta}<\ucst<2$,
as we shall prove in Corollary~\ref{cor:rip1},
\eqref{rip:strong} holds with high probability. Thus, the desired inequalities $0<\mathring{\lcst}^{\max}_{\beta}\le  \mathring{\ucst}^{\min}_{\beta}<2$ guarantees
the strong restricted isometry property for Gaussian random matrices.
%In fact, our analysis in Section~\ref{sec:beta} allows us to calculate/estimate $\alpha$ for given $\lcst$ and $\ucst$.

We have the following estimates on the quantities $\lcst^{\max}_\beta$, $\ucst^{\min}_\beta$ and $\mathring{\lcst}^{\max}_\beta$, $\mathring{\ucst}^{\min}_\beta$.
%$\lcst_\beta^{\max}$, $\ucst_\beta^{\min}$, $\mathring{\lcst}_\beta^{\max}$, and $\mathring{\ucst}_\beta^{\min}$.

\begin{theo}\label{thm:non}
Let $A$ be an $m\times \Ndim$ random matrix with i.i.d. entries obeying $\cN(0,1)$. For $0<\beta<1$,
\begin{align}
\frac{\pi}{6}(1-\beta)^2 \min\Big( \frac{3-2\beta}{4(1-\beta)},1\Big) &\le \lcst_\beta^{\max} \le \min\left(\frac{\pi}{2} \Big(\ln \frac{1}{\beta}\Big)^2, 1\right),
%\le \frac{\pi}{2\beta^2} (1-\beta)^2,
\label{eq:theta}\\
\max\Big(c_g^2 \ln \frac{2}{1-\beta}, \frac{\pi}{2}\beta^2\Big)
&\le \ucst_\beta^{\min} \le 2\ln \frac{e}{1-\beta},\label{eq:omega}
\end{align}
and
\begin{align}
\frac{\pi}{6}(1-\beta)^3 \min\Big( \frac{3-2\beta}{4(1-\beta)},1\Big) &\le \mathring{\lcst}_\beta^{\max} \le (1-\beta)\min\left( \frac{\pi}{2} \Big(\ln \frac{1}{\beta}\Big)^2,1\right),
%\le \frac{\pi}{2\beta^2} (1-\beta)^3,
\label{eq:theta:2}\\
(1-\beta)\max\Big(c_g^2 \ln \frac{2}{1-\beta}, \frac{\pi}{2}\beta^2\Big)
&\le \mathring{\ucst}_\beta^{\min} \le 2(1-\beta)\ln \frac{e}{1-\beta},\label{eq:omega:2}
\end{align}
where the absolute constant $c_g$ is defined in \eqref{cg}.
\end{theo}

Theorem~\ref{thm:non} establishes the strong restricted isometry property for Gaussian random matrices for all $\beta\in [0,1)$ by $(1-\beta)\ln \frac{e}{1-\beta}<1$ for any $\beta\in (0,1)$ and by \eqref{grm:normpreservation} for $\beta=0$.
As a direct consequence of Theorem~\ref{thm:non}, we have
\[
\frac{\pi}{6} \le  \frac{\lcst^{\max}_\beta}{(1-\beta)^2}
\le 2\pi(\ln 2)^2,\qquad
c_g^2 \le \frac{\ucst^{\min}_\beta}{\ln \frac{1}{1-\beta}}
\le 2+\frac{2}{\ln 2},
\qquad \forall\; 1/2\le \beta<1
\]
and
\[
\frac{\pi}{6} \le  \frac{\mathring{\lcst}^{\max}_\beta}{(1-\beta)^3}
\le 2\pi(\ln 2)^2,\qquad
c_g^2 \le \frac{\mathring{\ucst}^{\min}_\beta}{(1-\beta)\ln \frac{1}{1-\beta}}
\le 2+\frac{2}{\ln 2},
\qquad \forall\; 1/2\le \beta<1.
\]
Thus, up to multiplicative constants, our estimates in Theorem~\ref{thm:non} for $\lcst^{\max}_\beta$, $\ucst^{\min}_\beta$ and $\mathring{\lcst}^{\max}_\beta$, $\mathring{\ucst}^{\min}_\beta$
are optimal as $\beta\to 1^-$.

As an application of Theorem~\ref{thm:non} and our analysis in Section~\ref{sec:beta} for proving Theorem~\ref{thm:non},
%and other results in Section~\ref{sec:beta},
we have the following robustness properties of Johnson-Lindenstrauss lemma and restricted isometry property with a given erasure ratio $0<\beta<1$.

\begin{coro}\label{cor:JL1}
Let $0<\beta<1$ and $0<\alpha<\frac{\pi}{12}(1-\beta)^2 h_\beta$ with $h_\beta:=\min(\frac{3}{4}-\frac{1}{2}\beta,1-\beta)$.
Let $m,n,N\in \N$ such that $m\ge \frac{1}{1-\beta}$ and %$N<\frac{1}{2}+\sqrt{\frac{1}{4}+e^{\alpha m}}$.
$m>\frac{1}{\alpha}\ln \frac{1}{N(N-1)}$.
Let $A$ be an $m\times \Ndim$ Gaussian random matrix with i.i.d. entries obeying $\cN(0,1)$.
For any given $N$ points $p_1,\ldots, p_N\in \R^n$, with probability at least $1-N(N-1)e^{-\alpha m}>0$,
\begin{equation}\label{JL:beta}
\begin{split}
\lcst \|p_j-p_k\|^2 \le \tfrac{1}{m}\| A_Tp_j-A_T p_k\|^2\le &\ucst\|p_j-p_k\|^2,\qquad \\
&\forall\; 1\le j,k\le N\;\; \mbox{and}\;\;  T\subseteq\{1,\ldots,m\} \mbox{ with } \abs{T^c}\leq \beta m,
\end{split}
\end{equation}
where $\lcst,\ucst\in (0,\infty)$ are positive real numbers given by
\begin{equation}\label{theta:omega:JL}
\lcst:=\frac{\pi}{6}(1-\beta)^2h_\beta+2\alpha-2(1-\beta)
\sqrt{\pi \alpha h_\beta/3}, \quad \ucst:=\left(\sqrt{2(1-\beta-\tfrac{1}{m})\ln \tfrac{e}{1-\beta-\tfrac{1}{m}}}+\sqrt{2\alpha}\right)^2.
\end{equation}
\end{coro}

\begin{coro}\label{cor:rip1}
Let $0<\beta<1$ and $0<\alpha<\frac{\pi}{12}(1-\beta)^2 h_\beta$ with $h_\beta:=\min(\frac{3}{4}-\frac{1}{2}\beta,1-\beta)$.
Let $m,n,s\in \N$ such that $m\ge \frac{1}{1-\beta}$ and $s\ln \frac{24 en}{\epsilon s}<\alpha m-\ln 2$.
Let $A$ be an $m\times \Ndim$ Gaussian random matrix with i.i.d. entries obeying $\cN(0,1)$. For any $0<\epsilon<1$,
with probability at least $1-2(\frac{24en}{\epsilon s})^s e^{-\alpha m}>0$,
\begin{equation}
\lcst(1-\epsilon)\|v\|^2 \le \tfrac{1}{m}  \|A_T v\|^2\le \ucst(1+\epsilon)\|v\|^2, \quad \forall\; \mbox{$s$-sparse $v\in \R^{\Ndim}$ and}\;\; T\subseteq\{1,\ldots,m\} \mbox{ with } \abs{T^c}\leq \beta m,
\end{equation}
where the positive real numbers $\lcst$ and $\ucst$ are given in \eqref{theta:omega:JL}.
\end{coro}

It is of interest to extend the main results in this paper from Gaussian random matrices to other random matrices such as
sub-Gaussian matrices and circulant matrices (c.f. \cite{CJL}). If $A$ is the Bernoulli matrix, i.e., $\PP(a_{j,k}=1/\sqrt{m})=\PP(a_{j,k}=-1/\sqrt{m})=1/2$. Define 2-sparse vectors $v_1:=(1,1,0,\ldots,0)^\tp\in \R^n$ and $v_2:=(1,-1,0,\ldots,0)^\tp\in \R^n$. Then either
$\inf\left\{\frac{1}{m}\|A_T v_1\| \; : \; T\subseteq \{1,\ldots, m\}, |T^c|\le m/2\right\}=0$
or
$\sup\left\{\frac{1}{m}\|A_T v_2\| \; : \; T\subseteq \{1,\ldots, m\}, |T^c|\le m/2\right\}=0$.
That is, for any $\lcst>0$, either $\PP(\I_{[\lcst,\infty],1/2}(A,v_1))=0$
or $\PP(\I_{[\lcst,\infty],1/2}(A,v_2))=0$
for all $m\in \N$. As a consequence, the strong restricted isometry property for $\beta=1/2$ cannot hold for Bernoulli random matrices.
This shows that the results and study for sub-Gaussian random matrices will be essentially different to Gaussian random matrices. We shall study random matrices other than Gaussian random matrices elsewhere.

The structure of the paper is as follows. In Section~2, we shall provide some auxiliary results for the proofs of the main results  in later sections. In Section~3, we shall study the robustness properties of Gaussian random matrices with arbitrarily erased rows for small distortion rates $\epsilon\to 0^+$. In particular, we shall prove in Section~3 Theorem~\ref{thm:optimal} and a few other results related to Theorem~\ref{thm:optimal}.
In certain sense, we studied in Theorem~\ref{thm:optimal} the quantities
$\beta^{\max}_{\epsilon,\alpha}$ and $\mathring {\beta}^{\max}_{\epsilon,\alpha}$
 for the case of small erasure ratios $\epsilon\to 0^+$.
In Section~4, we shall study the robustness properties of Gaussian random matrices with arbitrarily erased rows for a given corruption/erasure ratio $0<\beta<1$. In particular, we are interested in the behavior of $\lcst^{\max}_\beta$, $\ucst^{\min}_\beta$ and
$\mathring{\lcst}^{\max}_\beta$, $\mathring{\ucst}^{\min}_\beta$
when $\beta\to 1^-$. We shall prove in Section~4 Theorem~\ref{thm:non} and other results related to Theorem~\ref{thm:non}.
%as well as the corresponding results when $\frac{1}{|T|}\|A_T x_0\|^2$ in \eqref{I:interval} is replaced by $\frac{1}{m}\|A_T x_0\|^2$.
We shall also show that our result leads to the establishment of the strong restricted isometry property for Gaussian random matrices.
As applications of the main results in this paper for dimensionality reduction, in Section~5 we shall prove Corollaries~\ref{cor:JL}, \ref{cor:rip} and Corollaries~\ref{cor:JL1}, \ref{cor:rip1}.

%%%%%%%%%%%%%%%%%%%%%%%%%%%%%%%%%%%%%%%%%%%%%%%%%%%%%%%%%%%%%%%%%%%%%
\section{Auxiliary Results}

In this section we provide some auxiliary results that will be used in later sections.
For $y=(y_1,\ldots,y_m)^\tp \in \mathbb{R}^m$,
we define $y_{(1)}, \ldots, y_{(m)}$ to be the nonincreasing rearrangements of $y_1, \ldots, y_m$ in terms of magnitudes such that $|y_{(1)}|\ge \cdots \ge |y_{(m)}|$.
Let $m\in \N$. For any $0\le \gamma\le 1$ such that $\gamma m$ is an integer, we define
\begin{equation}\label{T:gamma}
T_\gamma:=\{ T\subseteq \{1,\ldots,m\} \; :\; |T^c|=\gamma m\}.
\end{equation}

The following simple observation will facilitate our discussion in later sections.

\begin{lem}\label{lem:dec}
For $0\le \gamma\le \beta<1$ such that both $\gamma m$ and $\beta m$ are integers,
\[
\min_{T\in T_\beta} \left|\tfrac{1}{\abs{T}}\|A_T  x_0\|^2-1\right|\le \min_{T\in T_\gamma} \left|\tfrac{1}{\abs{T}}\|A_T  x_0\|^2-1\right| \le \max_{T\in T_\gamma} \left|\tfrac{1}{\abs{T}}\|A_T  x_0\|^2-1\right|
\le \max_{T\in T_\beta} \left|\tfrac{1}{\abs{T}}\|A_T  x_0\|^2-1\right|.
\]
\end{lem}

\begin{proof}
Let $k_\gamma:=\gamma m$ and $k_\beta:=\beta m$. By $0\le \gamma\le \beta$, we have $k_\gamma\le k_\beta$ and it follows from $|y_{(1)}|\ge \cdots \ge |y_{(m)}|$ that
\begin{align*}
\min_{T\in T_\beta} \frac{\|A_T  x_0\|^2}{|T|}
&=\frac{y_{(k_\beta+1)}^2+\cdots+y_{(m)}^2}{m-k_\beta}
\le \frac{y_{(k_\gamma+1)}^2+\cdots+y_{(m)}^2}{m-k_\gamma}=
\min_{T\in T_\gamma} \frac{\|A_T  x_0\|^2}{|T|}\\
&\le \max_{T\in T_\gamma} \frac{\|A_T  x_0\|^2}{|T|}
=\frac{y_{(1)}^2+\cdots+y_{(m-k_\gamma)}^2}{m-k_\gamma}
\le \frac{y_{(1)}^2+\cdots+y_{(m-k_\beta)}^2}{m-k_\beta}=
\max_{T\in T_\beta} \frac{\|A_T  x_0\|^2}{|T|} .
\end{align*}
Now the claim follows directly from the above inequalities.
\end{proof}

The following well-known concentration inequalities for the standard Gaussian/normal distribution (e.g., see \cite{probability}) will be used later.

\begin{theo} \label{thm:Lip}
Let $f: \mathbb{R}^m \rightarrow \R$ be a Lipschitz function with Lipschitz constant $1$ satisfying $|f(x)-f(y)| \le \|x-y\|$ for all $x,y\in \R^m$. For i.i.d. standard Gaussian/normal random variables $X_1,\ldots,X_m\sim \cN(0,1)$ and for all $\delta\ge 0$,
\begin{align}\label{eq: GC2}
\PP\bigl\{ f(X_1,\ldots, X_m) & < \delta+\E\left[ f(X_1,\ldots, X_m)\right]\bigl\}\,\, \ge\,\, 1-e^{-\delta^2/2},\\
\PP\bigl\{ f(X_1,\ldots, X_m) & > -\delta+\E\left[ f(X_1,\ldots, X_m)\right]\bigl\}\,\, \ge\,\, 1-e^{-\delta^2/2}.
\end{align}
\end{theo}

As an application of the above result, we have the following result (also c.f. \cite{VX}).

\begin{lem}\label{lem:Lip}
Let $y_1,\ldots,y_m$ be i.i.d. Gaussian/normal random variables obeying $\cN(0,1)$. Then for all nonempty subsets $S\subseteq \{1,\ldots, m\}$ and $\delta>0$,
\begin{equation}\label{est:mean:1}
\PP\left\{ \sqrt{\frac{1}{|S|} \sum_{j\in S} y_{(j)}^2}<\delta+\E\sqrt{\frac{1}{|S|} \sum_{j\in S} y_{(j)}^2}\right\}\ge 1-e^{-\delta^2 |S|/2}
\end{equation}
and
\begin{equation}\label{est:mean:2}
\PP\left\{ \sqrt{\frac{1}{|S|} \sum_{j\in S} y_{(j)}^2}>-\delta+\E\sqrt{\frac{1}{|S|} \sum_{j\in S} y_{(j)}^2}\right\}\ge 1-e^{-\delta^2 |S|/2}.
\end{equation}
\end{lem}

\begin{proof} Define
$F_S(x):=\sqrt{\sum_{j\in S} y_{(j)}^2}$. Then it is easy to observe that
\[
|F_S(x)-F_S(y)|^2\le \sum_{j\in S} (|x_{(j)}|-|y_{(j)}|)^2
\le \sum_{j=1}^m (|x_{(j)}|-|y_{(j)}|)^2=\|x\|^2+\|y\|^2-2\sum_{j=1}^m |x_{(j)}y_{(j)}|
\le \|x-y\|^2,
\]
where in the last step we used the rearrangement inequality $\sum_{j=1}^m |x_j y_j|\le \sum_{j=1}^m |x_{(j)}y_{(j)}|$.
Therefore, $F_S$ is a Lipschitz function with Lipschitz constant $1$.
 By Theorem~\ref{thm:Lip}, we have
\[
\PP\left\{ \sqrt{\frac{1}{|S|} \sum_{j\in S} y_{(j)}^2}<\delta+\E\sqrt{\frac{1}{|S|} \sum_{j\in S} y_{(j)}^2}\right\}=
\PP\left\{ \sqrt{\sum_{j\in S} y_{(j)}^2}<\delta \sqrt{|S|}+\E\sqrt{\sum_{j\in S} y_{(j)}^2}\right\}\ge 1-e^{-\delta^2 |S|/2}.
\]
The inequalities in \eqref{est:mean:2} can be proved similarly.
\end{proof}

The following result  extends \cite[Example~10]{min}. We provide a proof here by modifying the proof of \cite[Example~10]{min}.

\begin{lem}\label{lem:gmin}
Let $y_1,\ldots,y_m \in \R^m$ be i.i.d. Gaussian/normal random variables obeying $\cN(0,1)$
 and define $y_{(1)}, \ldots, y_{(m)}$ to be the nonincreasing rearrangements of $y_1, \ldots, y_m$ in terms of magnitudes such that $|y_{(1)}|\ge \cdots \ge |y_{(m)}|$.
\begin{enumerate}
\item[{\rm (i)}] For $1\leq j\leq m$ and  $1\le p\le 2$,
\begin{equation}\label{gmin:bound}
\sqrt{\frac{\pi}{2}}\frac{m+1-j}{m+1} \le \E \abs{y_{(j)}}
\quad \mbox{and}\quad
\E|y_{(j)}|^p\le C_p \sum_{\ell=j}^m \frac{1}{\ell}\le
C_p\left(\frac{1}{j}+\ln\frac{m}{j}\right),
\end{equation}
where $C_p$ is a positive constant (e.g., $C_1\le \sqrt{\frac{\pi}{2}}$, $C_2\le 2$) depending only on $p$ and given by
\begin{equation}\label{Cp}
C_p:=p\sup_{0<t<\infty} t^{p-1}\int_{t}^\infty e^{(t^2-s^2)/2}ds\le p \left(\frac{\pi}{2}\right)^{1-\tfrac{p}{2}}.
\end{equation}

\item[{\rm (ii)}] For $1\leq k\leq m$,
\begin{equation}\label{le:upperbound:new}
\E\left(\sqrt{\frac{1}{k}\sum_{j=1}^k y_{(j)}^2}\right) \le
%\min\Big\{
\sqrt{2\ln \frac{e m}{k}}.
%2\sqrt{\ln \frac{\sqrt{2}m}{k}} \Big \}.
\end{equation}

\item[{\rm (iii)}]
Let $0\le \gamma<1$ and $m\in \N$ such that $k:=\gamma m\in \N$. Then
\begin{equation}\label{est:small}
\sqrt{\frac{\pi}{6}} \sqrt{(1-\gamma)\frac{1-\gamma+\tfrac{1}{2m}}{1+\tfrac{1}{m}}} \le
\E \sqrt{\frac{1}{m-k}\sum_{j=k+1}^m y_{(j)}^2}\le \sqrt{2-\frac{2\gamma}{1-\gamma}\ln \frac{1+\tfrac{1}{m}}{\gamma+\tfrac{1}{m}}}.
\end{equation}
\end{enumerate}
\end{lem}

\begin{proof}
Define $u(t):=\sqrt{\frac{2}{\pi}} \int_t^{\infty} e^{-s^2/2} ds$. As shown in \cite[Example~10]{min},
\begin{align*}
\E|y_{(j)}|^p&=\int_0^\infty \PP\{ |y_{(j)}|>t^{1/p}\} dt
=\sum_{\ell=0}^{m-j} \binom{m}{\ell}\int_0^\infty (u(t^{1/p}))^{m-\ell} (1-u(t^{1/p}))^\ell dt\\
&=\sqrt{\frac{\pi}{2}} \sum_{\ell=0}^{m-j} \binom{m}{\ell}\int_0^\infty (u(t))^{m-\ell} (1-u(t))^\ell p t^{p-1} e^{t^2/2}(-du(t)).
\end{align*}
(i)  Since $e^{t^2/2}\ge 1$ for all $t\in \R$, for $p=1$, by a change of variable $x=u(t)$, as proved in \cite[Example~10]{min},
\[
E|y_{(j)}|\ge \sqrt{\frac{\pi}{2}} \sum_{\ell=0}^{m-j} \binom{m}{\ell}\int_0^1 x^{m-\ell} (1-x)^\ell dx
=\sqrt{\frac{\pi}{2}} \frac{m+1-j}{m+1}.
\]

For $1\le p\le 2$, by a change of variable $x=u(t)$, we deduce that
\begin{align*}
\E|y_{(j)}|^p&=\sum_{\ell=0}^{m-j} \int_0^\infty \sqrt{\frac{\pi}{2}} p t^{p-1}  e^{t^2/2} u(t)
\binom{m}{\ell} (u(t))^{m-\ell-1} (1-u(t))^\ell (-du(t))\\
&\le C_p \sum_{\ell=0}^{m-j} \binom{m}{\ell}\int_0^1 x^{m-\ell-1} (1-x)^\ell dx
=C_p \sum_{\ell=0}^{m-j} \binom{m}{\ell}\frac{(m-\ell-1)!\ell!}{m!}\\
&=C_p \sum_{\ell=0}^{m-j} \frac{1}{m-\ell}=C_p\left( \frac{1}{j}+\sum_{\ell=j+1}^m \frac{1}{\ell}\right)
\le C_p\left( \frac{1}{j}+\int_j^m \frac{1}{x} dx\right)
=C_p \left(\frac{1}{j}+\ln \frac{m}{j}\right).
\end{align*}
It is easy to prove that if $1\le p\le 2$, then $C_p<\infty$. Indeed, define
\[
f(t):=t^{p-1}\int_{t}^\infty e^{(t^2-s^2)/2}ds=t^{p-1}\int_{0}^\infty e^{-ts-s^2/2}ds\le t^{p-1}\int_0^\infty e^{-s^2/2}ds=
\sqrt{\tfrac{\pi}{2}} t^{p-1}.
\]
We also have
\[
f(t)=t^{p-1}\int_{0}^\infty e^{-ts-s^2/2}ds\le t^{p-1}\int_0^\infty e^{-ts} ds=t^{p-2}.
\]
Therefore, $C_p=p\sup_{0<t<\infty} f(t)\le p\sup_{0<t<\infty} \min(\sqrt{\tfrac{\pi}{2}} t^{p-1}, t^{p-2})=p(\tfrac{\pi}{2})^{1-\tfrac{p}{2}}<\infty$.

\begin{comment}
(ii)
Set $t_0:=c\sqrt{\ln \frac{m}{j}}$.
Then
\[
1-\sqrt{1-\exp(-t_0^2)}\leq  u(t_0) \leq 1-\sqrt{1-\exp(-t_0^2/2)}
\]
Using the Markov inequality, we have
\begin{align*}
\E\abs{y_{(j)}}& \ge c \sqrt{\ln \frac{m}{j}}\cdot \PP\{\abs{y_{(j)}}\geq  c \sqrt{\ln \frac{m}{j}} \}\\
&= c \sqrt{\ln \frac{m}{j}}\cdot \sum_{\ell=j}^m{m\choose \ell} u(t_0)^{\ell}(1-u(t_0))^{m-\ell}
\end{align*}
\end{comment}

(ii)
By \eqref{gmin:bound} with $p=2$, we have $C_2\le 2$ and
\[
\E\left(\sqrt{\frac{1}{k}\sum_{j=1}^k y_{(j)}^2}\right)\le \sqrt{ \frac{1}{k} \sum_{j=1}^k \E y_{(j)}^2}\le \sqrt{ \frac{2}{k} \sum_{j=1}^k \sum_{\ell=j}^m \frac{1}{\ell}}=
\sqrt{ \frac{2}{k} \left( k+ k\sum_{\ell=k+1}^m \frac{1}{\ell}\right)}\le
\sqrt{ 2+2\ln \frac{m}{k}}=\sqrt{2\ln \frac{em}{k}}
\]
by $\sum_{\ell=k+1}^m \frac{1}{\ell}\le \int_{k}^m \frac{1}{x}dx=\ln \frac{m}{k}$.
This proves \eqref{le:upperbound:new}.

\begin{comment}
On the other hand,
by the convexity of the exponential function, we have
\[
\exp\left(\E \left(\frac{1}{k}\sum_{j=1}^k\frac{y_{(j)}^2}{4}\right)\right)\le \frac{1}{k}\sum_{j=1}^k \E(\exp(\frac{y_{(j)}^2}{4}))\le \frac{1}{k} \sum_{j=1}^m \E(\exp(\frac{y_{(j)}^2}{4}))=\frac{1}{k} \sum_{j=1}^m \E(\exp(\frac{y_j^2}{4}))=
\frac{\sqrt{2}m}{k},
\]
since $\E(\exp(\frac{ y_j^2}{4}))=\sqrt{2}$ by $y_j\sim \cN(0,1)$.
Then we deduce from the above inequality that
\[
\E\left(\frac{1}{k}\sum_{j=1}^k y_{(j)}^2\right) \le 4\ln \frac{\sqrt{2}m}{k}.
\]
Therefore, by the Cauchy-Schwarz inequality, we conclude that
\[
\E\left(\sqrt{\frac{1}{k}\sum_{j=1}^k |y|_{(j)}^2}\right)
\le \sqrt{\E\left(\frac{1}{k}\sum_{j=1}^k |y|_{(j)}^2\right)} \le 2 \sqrt{\ln \frac{\sqrt{2}m}{k}}.
\]
This proves the right-hand inequality of \eqref{le:upperbound:new}.
\end{comment}

(iii)
By the Cauchy-Schwarz inequality, we have
\[
\sum_{j=k+1}^m \frac{m+1-j}{m+1}|y_{(j)}|\le \sqrt{\sum_{j=k+1}^m \left(\frac{m+1-j}{m+1}\right)^2}
\sqrt{\sum_{j=k+1}^m y_{(j)}^2}.
\]
Therefore, it follows from the first inequality in \eqref{gmin:bound} that
\begin{align*}
\E \sqrt{\frac{1}{m-k}\sum_{j=k+1}^m y_{(j)}^2}
&\ge \frac{1}{\sqrt{m-k}} \left(\sum_{j=k+1}^m \left(\frac{m+1-j}{m+1}\right)^2\right)^{-1/2}
\sum_{j=k+1}^m \frac{m+1-j}{m+1} \E |y_{(j)}|\\
&\ge \sqrt{\frac{\pi}{2}} \frac{1}{\sqrt{m-k}} \left(\sum_{j=k+1}^m \left(\frac{m+1-j}{m+1}\right)^2\right)^{-1/2}
\sum_{j=k+1}^m \left(\frac{m+1-j}{m+1}\right)^2\\
&=\sqrt{\frac{\pi}{2}} \sqrt{\frac{1}{m-k}\sum_{j=k+1}^m \left(\frac{m+1-j}{m+1}\right)^2}
=\sqrt{\frac{\pi}{2}}\sqrt{\frac{1}{(m-k)(m+1)^2}\sum_{j=1}^{m-k} j^2}\\
&=\sqrt{\frac{\pi}{2}} \sqrt{\frac{(m-k+1)(2m-2k+1)}{6(m+1)^2}}
\ge \sqrt{\frac{\pi}{6}} \sqrt{\frac{(m-k)(m-k+1/2)}{m (m+1)}},
\end{align*}
since $\frac{m-k+1}{m+1}\ge \frac{m-k}{m}$ by $k\ge 0$. By $k=\gamma m$, we proved the left-hand side of \eqref{est:small}.

On the other hand, it follows from the second inequality in \eqref{gmin:bound} with $p=2$ that
\begin{align*}
\E \sqrt{\frac{1}{m-k}\sum_{j=k+1}^m y_{(j)}^2}
&\le \sqrt{\frac{1}{m-k}\sum_{j=k+1}^m \E y_{(j)}^2}\le \sqrt{\frac{2}{m-k}\sum_{j=k+1}^m \sum_{\ell=j}^m \frac{1}{\ell}}=\sqrt{\frac{2}{m-k}\sum_{\ell=k+1}^m \sum_{j=k+1}^\ell \frac{1}{\ell}}\\
&=\sqrt{2-\frac{2k}{m-k}\sum_{\ell=k+1}^m \frac{1}{\ell}}\le \sqrt{2-\frac{2k}{m-k}\ln \frac{m+1}{k+1}},
\end{align*}
since $\sum_{\ell=k+1}^m \frac{1}{\ell}\ge \int_{k+1}^{m+1} \frac{1}{x} dx=\ln\frac{m+1}{k+1}$. By $k=\gamma m$, we proved the right-hand side of \eqref{est:small}.
\end{proof}

Let $W_0: [-e^{-1},\infty)\rightarrow [-1,\infty)$ and $W_{-1}: [-e^{-1},0)\rightarrow (-\infty, -1]$ be the principal and secondary real-valued Lambert $W$ functions
such that $W_0(x)e^{W_0(x)}=x$ for all $x\ge -e^{-1}$ and $W_{-1}(x)e^{W_{-1}(x)}=x$ for all $-e^{-1}\le x<0$ (see \cite{Lambert:1996}). Note that $W_0$ is an increasing function while $W_{-1}$ is a decreasing function.
%As inverse functions of $xe^x$ on different intervals, $W_{0}(xe^x)=x$ and $W_{0}'(x)=\frac{e^{-W_{0}(x)}}{1+W_{0}(x)}$ for all $x\in [-e^{-1},\infty)$, while
%$W_{-1}(xe^x)=x$ and $W_{-1}'(x)=\frac{e^{-W_{-1}(x)}}{1+W_{-1}(x)}$ for all $x\in [-e^{-1},0)$.

We shall also need the following auxiliary result in the proof of Theorem~\ref{thm:upperbound} of Section~\ref{sec:grmn}.

\begin{lem} \label{lem:ineq}
For any positive real number $c_g>0$,
\begin{equation}\label{epsg}
\epsilon_g:=\max_{x\ge 2} \frac{c_g^2 \ln(2x)-1}{x-1}
=\frac{c_g^2 \ln\tfrac{2}{\beta_g}-1}{\tfrac{1}{\beta_g}-1}
=\begin{cases}
c_g^2\ln 4-1>0, &\text{if $c_g^2\ln \tfrac{4}{\sqrt{e}}>1$,}\\
-c_g^2 W_0(-2e^{-1-1/c_g^2})>0, &\text{if $c_g^2\ln \tfrac{4}{\sqrt{e}}\le 1$,}
\end{cases}
\end{equation}
and
\begin{equation}\label{betag}
0<\beta_g\le \tfrac{1}{2}
\qquad \mbox{with}\qquad
\beta_g:=
\begin{cases} \tfrac{1}{2}, &\text{if $c_g^2 \ln\tfrac{4}{\sqrt{e}}>1$},\\
-W_0(-2e^{-1-1/c_g^2}), &\text{if $c_g^2 \ln\tfrac{4}{\sqrt{e}}\le 1$}.
\end{cases}
\end{equation}
%
%where $W_0: [-e^{-1},\infty)\rightarrow [-1,\infty)$ is the principal/upper branch of the Lambert $W$ function.
\end{lem}

\begin{proof} Let $t_g:=2e^{-1-1/c_g^2}$. If $0<c_g^2 \ln \tfrac{4}{\sqrt{e}} \le 1$, then $0<t_g <e^{-1}$ and therefore, both $W_0(-t_g)$ and $\beta_g$ are well defined. Since $W_0$ is an increasing function, it is also easy to prove that
\begin{equation}\label{betag:bound}
-W_0(-t_g)\le \tfrac{1}{2}\quad \Longleftrightarrow \quad t_g\le \tfrac{1}{2}e^{-1/2}\quad \Longleftrightarrow \quad  c_g^2 \ln\tfrac{4}{\sqrt{e}}\le 1.
\end{equation}
To prove \eqref{epsg}, we define $f(x):=\frac{c_g^2\ln(2x)-1}{x-1}$ for $x>1$.
Then
\[
f'(x)=\frac{g(x)}{x(x-1)^2} \quad \mbox{with}\quad g(x):=x+c_g^2(x-x\ln(2x)-1).
\]

By $g'(x)=1-c_g^2\ln(2x)$, the function $g$ increases on $(0,\frac{1}{et_g})$ and decreases on $(\frac{1}{et_g},\infty)$.
Note that $f'$ has the same sign as $g$ on $(1,\infty)$.
If $t_g \ge e^{-1}$, then $c_g^2\ln 2\ge 1$ and $g'(x)\le 0$ for all $x>1$. Hence, $f'(x)\le 0$ for all $x>1$ and $f$ decreases on $(1,\infty)$. Therefore, $\epsilon_g=\max_{x\ge 2} f(x)= f(2)=f(1/\beta_g)$, since $\beta_g=1/2$ by $c_g^2\ln \tfrac{4}{\sqrt{e}}> c_g^2\ln 2\ge 1$.
Since $g'$ has only one real root on $(0,\infty)$,
if $t_g<e^{-1}$,
then $g$ has at most two real roots on $(0,\infty)$ given by
$x_1:=\tfrac{1}{-W_{-1}(-t_g)}< 1 < \tfrac{1}{-W_{0}(-t_g)}=:x_2$.
Thus, $f$ decreases on $(x_2,\infty)$ and increases on $(1,x_2)$, from which we have $\epsilon_g=\max_{x\ge 2} f(x)=f(\max(2,x_2))=f(1/\beta_g)$, since $\max(2,x_2)=1/\beta_g$ by \eqref{betag:bound}.
\end{proof}

%%%%%%%%%%%%%%%%%%%%%%%%%%%%%%%%%%%%%%%%%%%%%%%%%%%%%%%%%%%%%%%%%%%%
\section{Gaussian Random Matrices under Arbitrary Erasure of Rows for Small $\epsilon$}
\label{sec:grmn}

In this section we study the robustness property of (normalized) Gaussian random matrices under arbitrary erasure of rows for small $0<\epsilon<1$. At the end of this section, we shall prove Theorem~\ref{thm:optimal}.

\subsection{A lower bound for $\beta^{\max}_{\epsilon,\alpha}$}

To provide a lower bound for $\beta^{\max}_{\epsilon,\alpha}$ in \eqref{betamax}, we first prove the following result.

\begin{lem}\label{lem:leth1}
Let $A$ be an $m\times \Ndim$ random matrix with i.i.d. entries obeying $\cN(0,1)$.
For $0<\alpha<1$ and $0<\epsilon\le \min(1,\frac{1-\sqrt{\alpha}}{\alpha/2})$,
if
\begin{equation}\label{cond:beta}
0<\beta\le \frac{1-\sqrt{\alpha}}{1+\epsilon} \epsilon \quad \mbox{and}\quad
0<\beta\ln\frac{e}{\beta}\le \frac{\epsilon}{2}\left(\sqrt{1-\sqrt{\alpha}}-\sqrt{\tfrac{\alpha\epsilon}{2}}\right)^2,
\end{equation}
then
\begin{equation}\label{eq:th1}
\PP(\I_{\epsilon,\beta} ) \ge 1-3e^{-\alpha (\epsilon^2/4-\epsilon^3/6) m}, \qquad \forall\, m\in \N.
\end{equation}
\end{lem}

\begin{proof}
Set $y:=Ax_0$. Since each entry in $A$ is an i.i.d. random variable obeying $\cN(0,1)$ and since $\|x_0\|=1$, a simple calculation leads to  $y_j\sim  \cN(0,1)$ for every $j=1, \ldots, m$ and all $y_1,\ldots,y_m$ are independent.

Let $m\in \N$ be arbitrary but fixed at this moment. Define $\gamma:= \lfloor\beta m\rfloor/m$. Then $0\le \gamma \le \beta$ and $\gamma m\in \N$.
If $\gamma=0$, then $\beta m<1$ and consequently $|T^c|\le \beta m$ implies that $T^c$ is the empty set. Therefore, if $\gamma=0$, then $\I_{\epsilon,\beta}=\I_{\epsilon,0}$ and the claim in \eqref{eq:th1} is trivially true by \eqref{grm:normpreservation}.
Thus, in the following discussion, we assume $\gamma>0$.

By Lemma~\ref{lem:dec}, we conclude that
\[
\PP(\I_{\epsilon,\beta})= \PP\left\{ 1-\epsilon\le \frac{\|y_T\|^2}{\abs{T}}\le 1+\epsilon\; \forall\, |T^c|\le \beta m \right\}
=\PP\left\{ 1-\epsilon\le \min_{T\in T_\gamma }\frac{\|y_T\|^2}{\abs{T}}\le
\max_{T\in T_\gamma}\frac{\|y_T\|^2}{\abs{T}}
\le 1+\epsilon\right\},
\]
where $T_\gamma$ is defined in \eqref{T:gamma}.
Recall that $y_{(1)},\ldots, y_{(m)}$ are nonincreasing rearrangements of $y_1,\ldots,y_m$ in terms of magnitude such that $\abs{y_{(1)}}\ge \abs{y_{(2)}}\ge \cdots \ge \abs{y_{(m)}}$.
Let $T\in T_\gamma$ and $k:=|T^c|=\gamma m$. Then $|T|=m-k$.
Since $\|y_T\|^2=\|y\|^2-\|y_{T^c}\|^2$, we observe that
\begin{equation}\label{yT}
\|y\|^2-(y_{(1)}^2+\cdots+y_{(k)}^2)\,\,\le\,\, \|y_T\|^2
\,\,\le\,\, \|y\|^2-(y_{(m-k+1)}^2+\cdots+y_{(m)}^2).
\end{equation}
Using the above inequalities and noting that $k=\gamma m$, we can easily deduce that
\begin{align*}
&\PP\left\{ 1-\epsilon\le \min_{T\in T_\gamma}\frac{\|y_T\|^2}{\abs{T}}\le
\max_{T\in T_\gamma}\frac{\|y_T\|^2}{\abs{T}}\le 1+\epsilon\right\}\\
&\qquad =\PP\left\{ 1-\epsilon\le \frac{\|y\|^2-(y_{(1)}^2+\cdots+y_{(k)}^2)}{m-k}
\le \frac{\|y\|^2-(y_{(m-k+1)}^2+\cdots+y_{(m)}^2)}{m-k}\le 1-\epsilon\right\}\\
&\qquad =\PP\left\{\frac{y_{(1)}^2+\cdots+y_{(k)}^2}{k}\le \frac{\|y\|^2}{k}-\frac{(1-\gamma)(1-\epsilon)}{\gamma}
\; \mbox{and}\; \frac{y_{(m-k+1)}^2+\cdots+y_{(m)}^2}{k}\ge \frac{\|y\|^2}{k}-\frac{(1-\gamma)(1+\epsilon)}{\gamma}\right\}\\
&\qquad \ge \PP(E_0\cap E_1\cap E_2)=1-\PP(E_0^c\cup E_1^c\cup E_2^c)\,\,\ge\,\, 1-\PP(E_0^c)-\PP(E_1^c)-\PP(E_2^c),
\end{align*}
where
\begin{align}
E_0&:=\left\{1-\sqrt{\alpha}\epsilon\le \frac{\|y\|^2}{m}\le 1+\sqrt{\alpha}\epsilon\right\},\label{E0}\\
E_1&:=
\left\{\frac{y_{(1)}^2+\cdots +y_{(k)}^2}{k} \le 1-\epsilon+\frac{1-\sqrt{\alpha}}{\gamma}\epsilon \right\},\label{E1}\\
E_2&:=
%\left\{\frac{y_{(m)}^2+\cdots +y_{(m-k+1)}^2}{k} \ge 1-\frac{m(1-\sqrt{\alpha})-k}{k}\epsilon \right\}=
\left\{\frac{y_{(m)}^2+\cdots +y_{(m-k+1)}^2}{k} \ge 1+\epsilon-\frac{1-\sqrt{\alpha}}{\gamma}\epsilon \right\}.\label{E2}
\end{align}
Since $E_0=\{ |\tfrac{1}{m} \|y\|^2-1|\le \sqrt{\alpha}\epsilon\}$, it follows directly from \eqref{grm:normpreservation} that
$\PP(E_0)\,\,\ge\,\, 1-2e^{-(\alpha \epsilon^2/4-\alpha^{3/2} \epsilon^3/6) m}$.
Thus, by $0<\alpha<1$, we have
\[
\PP(E_0^c)\,\,\le\,\, 2e^{-(\alpha \epsilon^2/4-\alpha^{3/2}\epsilon^3/6) m}=
2e^{-\alpha (\epsilon^2/4-\epsilon^3/6) m} e^{-(1-\sqrt{\alpha})\alpha \epsilon^3 m/6}\le 2e^{-\alpha (\epsilon^2/4-\epsilon^3/6) m}.
\]
Next we estimate $\PP(E_1)$ and $\PP(E_2)$.
By \eqref{le:upperbound:new} of Lemma \ref{lem:gmin} and noting that $k=\gamma m$,
we have
\[
\E\sqrt{\frac{1}{k}\sum_{j=1}^k y_{(j)}^2}\le \sqrt{2\ln \frac{e}{\gamma}}.
\]
For $\delta>0$, it follows from \eqref{est:mean:1} of Lemma~\ref{lem:Lip} with $S=\{1,\ldots, k\}$ that
\begin{equation}\label{eq:pr2}
\PP\left\{\sqrt{\frac{1}{k} \sum_{j=1}^k y_{(j)}^2 }\le \delta +\sqrt{2\ln \frac{e}{\gamma}}\right\}\ge \PP\left\{\sqrt{\frac{1}{k} \sum_{j=1}^k y_{(j)}^2 }\le \delta +\E \sqrt{\frac{1}{k} \sum_{j=1}^k y_{(j)}^2 }\right\}
\ge 1-e^{-\delta^2 \gamma m/2}.
\end{equation}
Take
\begin{equation}\label{eq1:delta}
\delta:=\sqrt{1-\epsilon+\frac{1-\sqrt{\alpha}}{\gamma}\epsilon}- \sqrt{2\ln \frac{e}{\gamma}}.
\end{equation}
We claim that
\begin{equation}\label{cond:beta:alpha}
\delta\,\, \ge\,\, \sqrt{\frac{\alpha \epsilon^2}{2\gamma}}\,\,>\,\,0.
\end{equation}
Then it follows from \eqref{eq:pr2} and the above inequality that
\[
\PP(E_1^c)=\PP\left\{ \sqrt{\frac{1}{k}\sum_{j=1}^k y_{(j)}^2}> \sqrt{1-\epsilon+\frac{1-\sqrt{\alpha}}{\gamma}\epsilon} \right\} \le e^{-\delta ^2 \gamma m
/2}\le e^{-\alpha\epsilon^2 m/4}\le
e^{-\alpha (\epsilon^2/4-\epsilon^3/6) m}.
\]

On the other hand, by the first inequality in \eqref{cond:beta} and the fact that $0<\gamma\leq \beta$, we have
\[
1+\epsilon-\frac{1-\sqrt{\alpha}}\gamma\epsilon \le 1+\epsilon-\frac{1-\sqrt{\alpha}}\beta\epsilon
\le 1+\epsilon-(1+\epsilon)=0,
\]
which yields
\[
\PP(E_2)=\PP\left\{\frac{y_{(m)}^2+\cdots +y_{(m-k+1)}^2}{k} \ge 1+\epsilon-\frac{1-\sqrt{\alpha}}{\gamma}\epsilon\right\}\ge
\PP\{ y_{(m)}^2+\cdots +y_{(m-k+1)}^2\ge 0\}=1.
\]
That is, $\PP(E_2^c)=0$. Putting all estimates together, we complete the proof of \eqref{eq:th1}.

We now prove \eqref{cond:beta:alpha}.
By our assumption $0<\epsilon\le \frac{1-\sqrt{\alpha}}{\alpha/2}$, we observe that
\begin{equation}\label{fepsalpha}
f_{\epsilon,\alpha}\ge 0 \qquad \mbox{with}\quad f_{\epsilon,\alpha}:=\sqrt{\frac{\epsilon}{2}}\left(\sqrt{1-\sqrt{\alpha}}-\sqrt{\tfrac{\alpha \epsilon}{2}}\right).
\end{equation}
Now it is straightforward to check that
\begin{equation}\label{heps}
2x^2+2\sqrt{\alpha \epsilon^2}x\,\,\le\,\,
2f_{\epsilon,\alpha}^2+2\sqrt{\alpha \epsilon^2}f_{\epsilon,\alpha}=
(1-\sqrt{\alpha})\epsilon-\tfrac{\alpha \epsilon^2}{2}, \qquad \text{for all } x\in [0, f_{\epsilon,\alpha}].
\end{equation}
Note that the function $x\ln \frac{e}{x}$ is an increasing function on the interval $(0,1]$.
%We now prove that \eqref{cond:beta} implies $0<\beta<\tfrac{\sqrt{2}}{e}$.
%Indeed, \eqref{cond:beta} implies $0<\beta\le (1-\sqrt{\alpha})\frac{\epsilon}{1+\epsilon}\le (1-\sqrt{\alpha})/2\le 1/2$ and $0<\beta\ln\tfrac{\sqrt{2}}{\beta}\le \tfrac{\epsilon}{4}(1-\sqrt{\alpha})\le \tfrac{1}{4}$.
%By $0<\beta\le 1/2$ and $\beta\ln \tfrac{\sqrt{2}}{\beta}\le 1/4$, we have
%\[
%0<\beta\ln(2\sqrt{2})\le \beta\ln \tfrac{\sqrt{2}}{\beta}\le \tfrac{1}{4} \Longrightarrow
%\beta\le \tfrac{1}{4\ln(2\sqrt{2})}=\tfrac{1}{6\ln 2} <\tfrac{\sqrt{2}}{e}.
%\]
Since $0<\gamma\le \beta<1$, by the second inequality in \eqref{cond:beta}, we have
\[
0<\sqrt{\gamma \ln\tfrac{e}{\gamma}}\le \sqrt{\beta \ln\tfrac{e}{\beta}} \le f_{\epsilon,\alpha}.
\]
Consequently, plugging $x=\sqrt{\gamma\ln\tfrac{e}{\gamma}}$ into the inequality in \eqref{heps}, we deduce that
\[
2\gamma \ln \tfrac{e}{\gamma}+2\sqrt{\alpha \epsilon^2} \sqrt{\gamma \ln \tfrac{e}{\gamma}}\le (1-\sqrt{\alpha})\epsilon-\tfrac{\alpha \epsilon^2}{2}.
\]
Since $\gamma>0$, dividing $\gamma$ on both sides of the above inequality, we obtain
\[
2\ln \tfrac{e}{\gamma}+2\sqrt{\tfrac{\alpha \epsilon^2}{2\gamma}} \sqrt{2\ln \tfrac{e}{\gamma}}\le \tfrac{1-\sqrt{\alpha}}{\gamma}\epsilon-\tfrac{\alpha \epsilon^2}{2\gamma}.
\]
Hence, by $1-\epsilon \ge 0$, we conclude that
\[
\left(\sqrt{2\ln \tfrac{e}{\gamma}}+\sqrt{\tfrac{\alpha \epsilon^2}{2\gamma}}\right)^2=
2\ln \frac{e}{\gamma}+2\sqrt{\tfrac{\alpha \epsilon^2}{2\gamma}} \sqrt{2\ln \tfrac{e}{\gamma}}+\frac{\alpha \epsilon^2}{2\gamma}
\le \frac{1-\sqrt{\alpha}}{\gamma}\epsilon
\le 1-\epsilon+\frac{1-\sqrt{\alpha}}{\gamma}\epsilon.
\]
That is, we proved
\[
\sqrt{2\ln \frac{e}{\gamma}}+\sqrt{\frac{\alpha \epsilon^2}{2\gamma}}
\le \sqrt{1-\epsilon+\frac{1-\sqrt{\alpha}}{\gamma}\epsilon},
\]
which is simply the inequality in \eqref{cond:beta:alpha}.
\end{proof}

\begin{comment}
By essentially the same argument (but more complicated expressions) as in the proof of
Lemma~\ref{lem:leth1},
the range of $\beta$ in \eqref{cond:beta} of
Lemma~\ref{lem:leth1} can be further improved if the sharper upper bound $\sqrt{-W_{-1}(-e^{-1}\gamma^2}$ is used instead of the simple estimate $2\sqrt{\ln \frac{\sqrt{2}}{\gamma}}$ in \eqref{le:upperbound} with $t=1/4$.
\end{comment}

The following result establishes a lower bound for $\beta^{\max}_{\epsilon,\alpha}$ in \eqref{betamax}.

\begin{theo}\label{thm:lowerbound}
Let $A$ be an $m\times \Ndim$ random matrix with i.i.d. entries obeying $\cN(0,1)$.
For $0<\alpha<1$ and $0<\epsilon \le \min(1,\frac{1-\sqrt{\alpha}}{4\alpha})$,
if
\begin{equation}\label{beta:epsalpha}
0<\beta\le \frac{(1-\sqrt{\alpha})\epsilon}{16\ln \tfrac{4}{(1-\sqrt{\alpha})\epsilon}},
\end{equation}
then \eqref{eq:th1} holds.
Consequently,
\begin{equation}\label{beta:eps}
\beta^{\max}_{\epsilon,\alpha}>\frac{\epsilon}{\ln \tfrac{1}{\epsilon}}\frac{1-\sqrt{\alpha}}{32},\qquad\forall\, 0<\epsilon\le \frac{1-\sqrt{\alpha}}{4}.
\end{equation}
\end{theo}

\begin{proof}
We first show that \eqref{beta:epsalpha} implies \eqref{cond:beta} in Lemma~\ref{lem:leth1}.
Since $0<\epsilon\le1$ and $0<\alpha<1$, we have $0<(1-\sqrt{\alpha})\epsilon\le 1$.
The first inequality in \eqref{cond:beta} follows from \eqref{beta:epsalpha}, since
\[
0<\beta\le t_{\epsilon,\alpha}:=\frac{(1-\sqrt{\alpha})\epsilon}{16\ln \tfrac{4}{(1-\sqrt{\alpha})\epsilon}}\le \frac{(1-\sqrt{\alpha})\epsilon}{16\ln 4}< \frac{(1-\sqrt{\alpha})\epsilon}{2}\le \frac{(1-\sqrt{\alpha})\epsilon}{1+\epsilon}.
\]
Let $f_{\epsilon,\alpha}$ be defined in \eqref{fepsalpha}. By $0<\epsilon\le \frac{1-\sqrt{\alpha}}{4\alpha}$, we have $\frac{\alpha \epsilon}{2}\le \frac{1-\sqrt{\alpha}}{8}$ and hence, $f_{\epsilon,\alpha}>0$ and
\begin{equation}\label{est:fepsalpha}
f_{\epsilon,\alpha}^2=\frac{\epsilon}{2}\left( \sqrt{1-\sqrt{\alpha}}-\sqrt{\frac{\alpha\epsilon}{2}}\right)^2
\ge \frac{\epsilon}{2}\left( \sqrt{1-\sqrt{\alpha}}-\sqrt{\tfrac{1-\sqrt{\alpha}}{8}}\right)^2
=\frac{9-4\sqrt{2}}{16}(1-\sqrt{\alpha})\epsilon.
\end{equation}
A basic calculation shows that
\[
\ln(4ez) < (8-4\sqrt{2})z, \qquad \forall\; z\ge \ln 4,
\]
from which it is straightforward to deduce, by setting $z=\ln(1/x)$, that
\[
\frac{x}{4\ln(1/x)}\ln\frac{4e\ln(1/x)}{x}< \frac{9-4\sqrt{2}}{4}x, \qquad \forall\, 0<x\le \frac{1}{4}.
\]
Plugging $x:=\frac{(1-\sqrt{\alpha})\epsilon}{4}\le 1/4$ into the above inequality,
by \eqref{est:fepsalpha}, we conclude that
\[
t_{\epsilon,\alpha} \ln\frac{e}{t_{\epsilon,\alpha}}=\frac{x}{4\ln (1/x)}\ln\frac{4e\ln(1/x)}{x}< \frac{9-4\sqrt{2}}{4}x=\frac{9-4\sqrt{2}}{16}(1-\sqrt{\alpha})\epsilon\le f_{\epsilon,\alpha}^2.
\]
Since $\beta\ln\tfrac{e}{\beta}$ is an increasing function on $(0,1]$,
% and $t_{\epsilon,\alpha}\le \frac{1}{28\ln 4}<\tfrac{\sqrt{2}}{e}$. Consequently,
the second inequality in \eqref{cond:beta} follows from \eqref{beta:epsalpha} by noting that
$0<\beta\ln \frac{e}{\beta}\le t_{\epsilon,\alpha} \ln \frac{e}{t_{\epsilon,\alpha}}<f_{\epsilon,\alpha}^2$.

Since $t_{\epsilon,\alpha}<\frac{(1-\sqrt{\alpha})\epsilon}{1+\epsilon}$ and $t_{\epsilon,\alpha} \ln \frac{e}{t_{\epsilon,\alpha}}<f_{\epsilon,\alpha}^2$, there must exist $\delta>0$ such that \eqref{cond:beta} holds for all $0<\beta\le t_{\epsilon,\alpha}+\delta$.
Note that $\epsilon\le \frac{1-\sqrt{\alpha}}{4}< \frac{1-\sqrt{\alpha}}{4\alpha}$ by $0<\alpha<1$. We have $\ln\tfrac{1}{\epsilon}\ge \ln \tfrac{4}{1-\sqrt{\alpha}}$ and therefore,
\[
\beta^{\max}_{\epsilon,\alpha}\ge t_{\epsilon,\alpha}+\delta
>\frac{(1-\sqrt{\alpha})\epsilon}{16(\ln \tfrac{1}{\epsilon}+
\tfrac{4}{1-\sqrt{\alpha}})}\ge \frac{(1-\sqrt{\alpha})\epsilon}{32 \ln \tfrac{1}{\epsilon}}.
\]
This proves \eqref{beta:eps}.
%by noting that$\epsilon\le \frac{1-\sqrt{\alpha}}{4}\le \frac{1-\sqrt{\alpha}}{4\alpha}$.
\end{proof}

\subsection{An upper bound for $\beta^{\max}_{\epsilon,\alpha}$}

We now show that the order $\beta^{\max}_{\epsilon,\alpha}=O(\epsilon/\ln \tfrac{1}{\epsilon})$ for small $\epsilon$ given in Theorem~\ref{thm:lowerbound} is optimal.
To do so, we recall a well-known inequality on order statistics (e.g., see item (ii) of \cite[Example~10]{min} or see \cite[Lemma~3.3.1]{CGLP}):
there exists an absolute positive constant $c_g$ (depending only on the Gaussian/normal distribution $\cN(0,1)$) such that
\begin{equation}\label{cg}
c_g \sqrt{ \ln \tfrac{2m}{j}}\leq \E \abs{y_{(j)}},\qquad \forall\; 1\le j\le m/2.
\end{equation}

We first estimate the quantity $\PP(\I_{\epsilon,\beta})$.

\begin{lem}\label{lem:beta}
Let $m\in \N$ and $0<\beta\le 1/2$ such that $\beta m\in \N$. Let $\tilde{\epsilon}>0$ and $0<\epsilon\le 1$. If
\begin{equation} \label{delta:beta}
\delta:=c_g\sqrt{\beta\ln \frac{2}{\beta}}
-\sqrt{(1-\epsilon)\beta+\epsilon+\tilde{\epsilon}}>0,
\end{equation}
then
\begin{equation} \label{est:PIepsgamma}
\PP(\I_{\epsilon,\beta})\le e^{-\tilde{\epsilon}^2 m/4}+e^{-\delta^2 m/2}.
\end{equation}
\end{lem}

\begin{proof} Let $y:=Ax_0$ and $k:=\beta m$. By Lemma~\ref{lem:dec} and \eqref{yT}, we have
\begin{align*}
\PP(\I_{\epsilon,\beta})&=\PP\left\{ \max_{T\in T_\beta}|\tfrac{1}{|T|} \|y_T\|^2-1|\le \epsilon\right\}\le \PP\left\{ \min_{T\in T_\beta} \|y_T\|^2\ge (1-\epsilon)(1-\beta)m\right \}\\
&=\PP\left\{ \|y\|^2-(y_{(1)}^2+\cdots+y_{(k)}^2)\ge (1-\epsilon)(1-\beta)m\right \}\\
&\le \PP\{ \tfrac{1}{m} \|y\|^2>1+\tilde{\epsilon}\}+\PP\{y_{(1)}^2+\cdots+y_{(k)}^2\le \eta m\},
\end{align*}
where $\eta:=(1-\epsilon)\beta+\epsilon+\tilde{\epsilon}>0$ since $0<\epsilon\le 1$.
It follows directly from \eqref{eq:epsilon:0} that
\begin{equation}\label{tildeeps}
\PP\{ \tfrac{1}{m} \|y\|^2>1+\tilde{\epsilon}\}\le e^{-\tilde{\epsilon}^2m/4}.
\end{equation}
On the other hand,
\[
\PP\left\{y_{(1)}^2+\cdots+y_{(k)}^2\le \eta m\right\}=
\PP\left\{\sqrt{\frac{1}{k}\sum_{j=1}^k y_{(j)}^2}-\E\sqrt{\frac{1}{k}\sum_{j=1}^k y_{(j)}^2}
\le \sqrt{\frac{\eta}{\beta}}-\E\sqrt{\frac{1}{k}\sum_{j=1}^k y_{(j)}^2}\right\}.
\]
Since $k=\beta m\le m/2$, it follows from the inequality in \eqref{cg} and $|y_{(1)}|\ge \cdots \ge |y_{(k)}|$ that
\[
\E\sqrt{\frac{1}{k}\sum_{j=1}^k y_{(j)}^2}\ge \E(|y_{(k)}|)\ge c_g \sqrt{\ln \frac{2m}{k}}=c_g\sqrt{\ln\frac{2}{\beta}}.
\]
Therefore,
\begin{align*}
\PP\left\{y_{(1)}^2+\cdots+y_{(k)}^2\le \eta m\right\}
&=
\PP\left\{\sqrt{\frac{1}{k}\sum_{j=1}^k y_{(j)}^2}-\E\sqrt{\frac{1}{k}\sum_{j=1}^k y_{(j)}^2}
\le \sqrt{\frac{\eta}{\beta}}-\E\sqrt{\frac{1}{k}\sum_{j=1}^k y_{(j)}^2}\right\}\\
&\le \PP\left\{\sqrt{\frac{1}{k}\sum_{j=1}^k y_{(j)}^2}-\E\sqrt{\frac{1}{k}\sum_{j=1}^k y_{(j)}^2}\le \sqrt{\frac{\eta}{\beta}}-c_g\sqrt{\ln \frac{2}{\beta}}\right\}\\
&=\PP\left\{\sqrt{\frac{1}{k}\sum_{j=1}^k y_{(j)}^2}-\E\sqrt{\frac{1}{k}\sum_{j=1}^k y_{(j)}^2}\le -\frac{\delta}{\sqrt{\beta}} \right\}.
\end{align*}
By our assumption $\delta>0$,
applying \eqref{est:mean:2} in Lemma~\ref{lem:Lip} with $S=\{1,\ldots, k\}$, we get
\[
\PP\left\{y_{(1)}^2+\cdots+y_{(k)}^2\le \eta m\right\}\le
\PP\left\{\sqrt{\frac{1}{k}\sum_{j=1}^k y_{(j)}^2}-\E\sqrt{\frac{1}{k}\sum_{j=1}^k y_{(j)}^2}\le -\frac{\delta}{\sqrt{\beta}} \right\}\le
e^{-\delta^2 k/(2\beta)}=e^{-\delta^2 m/2}.
\]
Combining the above inequality with \eqref{tildeeps}, we conclude that \eqref{est:PIepsgamma} holds.
\end{proof}

To provide an upper bound for $\beta^{\max}_{\epsilon,\alpha}$ in \eqref{betamax}, here we introduce a related quantity. For $\epsilon>0$, $\alpha>0$ and $C>0$, we define
\begin{equation}\label{betamax:2}
\beta^{\max}_{\epsilon,\alpha,C}:=\sup\{0<\beta<1\; :\; \PP(\I_{\epsilon,\beta}) \ge 1-C e^{-\alpha (\epsilon^2/4-\epsilon^3/6) m}\; \text{for sufficiently large}\; m\in \N\}.
\end{equation}
If the above set on the right-hand side of \eqref{betamax} is empty, then we simply define $\beta^{\max}_{\epsilon,\alpha,C}:=0$. Trivially, $\beta^{\max}_{\epsilon,\alpha}\le \beta^{\max}_{\epsilon,\alpha,3}$ for all $\epsilon>0$ and $\alpha>0$.

\begin{theo}\label{thm:upperbound}
Let $c_g$ and $\epsilon_g$ be the absolute positive constants defined in \eqref{cg} and \eqref{epsg}. Then
\begin{equation}\label{eq:upperbound}
\beta^{\max}_{\epsilon,\alpha, C} < \frac{(1+\epsilon_g^{-1})\epsilon}{c_g^2 \ln \tfrac{2\epsilon_g}{\epsilon}}, \qquad\forall\, 0<\epsilon<\min(1,\epsilon_g), \alpha>0, C>0.
\end{equation}
\end{theo}

\begin{proof}
If $\beta^{\max}_{\epsilon,\alpha,C}=0$, the claim is trivially true. Hence, we assume $\beta^{\max}_{\epsilon,\alpha,C}>0$. We first prove that
\begin{equation}\label{cond:beta:max}
f_\epsilon(\beta):=c_g \sqrt{\beta\ln \tfrac{2}{\beta}}-\sqrt{(1-\epsilon)\beta+\epsilon}\le 0 \qquad \forall\; 0<\beta<\min(\tfrac{1}{2}, \beta^{\max}_{\epsilon,\alpha,C}).
\end{equation}
By the continuity of the function $f_{\epsilon}$, it suffices to prove \eqref{cond:beta:max} under the extra assumption that $\beta$ is a rational number.
Suppose that \eqref{cond:beta:max} fails for some rational number $\beta$ such that $0<\beta<\min(\tfrac{1}{2}, \beta^{\max}_{\epsilon,\alpha,C})$. Then there exists $\tilde{\epsilon}>0$ such that
\[
\delta=c_g\sqrt{\beta\ln \frac{2}{\beta}}
-\sqrt{(1-\epsilon)\beta+\epsilon+\tilde{\epsilon}}>0,
\]
 where $\delta$ is also defined in \eqref{delta:beta}.
Consequently, by Lemma~\ref{lem:beta},
\begin{equation}\label{eq:betaupper1}
\PP(\I_{\epsilon,\beta})\le e^{-\tilde{\epsilon}^2 m/4}+e^{-\delta^2 m/2}
\end{equation}
provided $\beta m\in \N$.
On the other hand, by the definition of $\beta^{\max}_{\epsilon,\alpha,C}$ and $0< \beta<\beta^{\max}_{\epsilon,\alpha,C}$,
\begin{equation}\label{eq:betaupper2}
\PP(\I_{\epsilon,\beta})\ge 1-Ce^{-\alpha (\epsilon^2/4-\epsilon^3/6) m}\quad \mbox{for sufficiently large}\; m\in \N.
\end{equation}
Consequently, combining \eqref{eq:betaupper1} and \eqref{eq:betaupper2}, we have
\[
1-Ce^{-\alpha (\epsilon^2/4-\epsilon^3/6) m}\le \PP(\I_{\epsilon,\beta})\le e^{-\tilde{\epsilon}^2 m/4}+e^{-\delta^2 m/2}
\]
for sufficiently large $m\in \N$ satisfying $\beta m\in \N$.
Since $\beta$ is a rational number, there are infinitely many sufficiently large $m\in \N$ satisfying $\beta m\in \N$. Letting such $m$ go to $\infty$, we deduce from the above inequality that
\[
1\le Ce^{-\alpha (\epsilon^2/4-\epsilon^3/6) m}+e^{-\tilde{\epsilon}^2 m/4}+e^{-\delta^2 m/2}\to 0,
\]
which is a contradiction. Therefore, \eqref{cond:beta:max} must hold.

Define a function
\[
F_\epsilon(\beta):=c_g^2 \beta \ln \tfrac{2}{\beta}-((1-\epsilon)\beta+\epsilon),\qquad \beta>0.
\]
Then it is trivial to see that $f_\epsilon(\beta)$ and $F_\epsilon(\beta)$ have the same sign on the interval $\beta\in (0,1)$.
As a direct consequence of \eqref{cond:beta:max}, it is straightforward to see that
\begin{equation}\label{est:betamax}
\min(\tfrac{1}{2},\beta^{\max}_{\epsilon,\alpha,C})\le \inf\{ 0<\beta <1/2 \; : \; F_{\epsilon}(\beta)>0\}=:\beta_\epsilon.
\end{equation}
If the above set on the right-hand side is empty, then we simply define $\beta_\epsilon=1/2$.
Let $\beta_g$ and $\epsilon_g$ be defined as in Lemma~\ref{lem:ineq}.
Since $0<\beta_g\le 1/2$ and $\epsilon_g>0$, we deduce that
\[
F_\epsilon(\beta_g)=c_g^2 \beta_g\ln\tfrac{2}{\beta_g}-((1-\epsilon)\beta_g+\epsilon)
=(1-\beta_g)(\epsilon_g-\epsilon)>0,\qquad \forall\,
0<\epsilon<\epsilon_g,
\]
where we used $c_g^2\ln \tfrac{2}{\beta_g}-1=\epsilon_g(\frac{1}{\beta_g}-1)$ by \eqref{epsg}.
Since $\lim_{\beta\to 0^+}F_\epsilon(\beta)=-\epsilon<0$, $F_\epsilon$ must have a real root inside the interval $(0,\beta_g)$. Hence, $0<\beta_\epsilon<\beta_g\le 1/2$ and $F_\epsilon(\beta_\epsilon)=0$. For $0<\epsilon<\epsilon_g$,
\[
\epsilon=F_\epsilon(\beta_\epsilon)+\epsilon
=(c_g^2\ln\tfrac{2}{\beta_\epsilon}-1+\epsilon)\beta_\epsilon \ge (c_g^2\ln\tfrac{2}{\beta_g}-1)\beta_\epsilon=\epsilon_g(\tfrac{1}{\beta_g}-1) \beta_\epsilon \ge \epsilon_g \beta_\epsilon,
\]
by $0<\beta_g\le 1/2$.
It follows from the above inequality that $\beta_\epsilon\le \tfrac{\epsilon}{\epsilon_g}$ for all $0<\epsilon<\epsilon_g$. Hence, for $0<\epsilon<\epsilon_g$, it follows from $F_\epsilon(\beta_\epsilon)=0$
and $0<\beta_\epsilon\le \tfrac{\epsilon}{\epsilon_g}$ that
\[
\beta_\epsilon=\frac{\epsilon+(1-\epsilon)\beta_\epsilon}{c_g^2 \ln \tfrac{2}{\beta_\epsilon}} \le \frac{\epsilon+(1-\epsilon)\tfrac{\epsilon}{\epsilon_g}}{c_g^2\ln \tfrac{2\epsilon_g}{\epsilon}}
<\frac{(1+\epsilon_g^{-1})\epsilon}{c_g^2\ln \tfrac{2\epsilon_g}{\epsilon}}.
\]
By $0<\beta_\epsilon<\beta_g\le 1/2$, we conclude that
\begin{equation}\label{est:beta:max}
\beta_{\epsilon,\alpha,C}^{\max}=\min(\tfrac{1}{2},\beta_{\epsilon,\alpha,C}^{\max})\le \beta_\epsilon < \frac{(1+\epsilon_g^{-1})\epsilon}{c_g^2\ln \tfrac{2\epsilon_g}{\epsilon}}, \qquad \forall\, 0<\epsilon<\min(1,\epsilon_g).
\end{equation}
This proves \eqref{eq:upperbound}.
\end{proof}

\subsection{An estimate for $\mathring{\beta}^{\max}_{\epsilon,\alpha}$}

We now study the relation of the quantity $\mathring{\beta}^{\max}_{\epsilon,\alpha}$ in \eqref{mbetamax} using  $\mathring{\I}_{\epsilon,\beta}$ in \eqref{mI:epsbeta} with a uniform normalization factor $\frac{1}{m}$.

Similar to Lemma~\ref{lem:beta} and Theorem~\ref{thm:lowerbound}, we have the following result on the lower bound of $\mathring{\beta}^{\max}_{\epsilon,\alpha}$.

\begin{theo}\label{thm:lowerbound:0}
Let $A$ be an $m\times \Ndim$ random matrix with i.i.d. entries obeying $\cN(0,1)$.
For $0<\alpha<1$ and $0<\epsilon \le \min(1,\frac{1-\sqrt{\alpha}}{\alpha/2})$, if $0<\beta<1$ satisfies the second inequality in \eqref{cond:beta}, i.e.,
\begin{equation}\label{eq:312}
0<\beta\ln\frac{e}{\beta}\le \frac{\epsilon}{2}\left(\sqrt{1-\sqrt{\alpha}}-\sqrt{\tfrac{\alpha\epsilon}{2}}\right)^2,
\end{equation}
then
\begin{equation}\label{eq:th1:0}
\PP(\mathring{\I}_{\epsilon,\beta} ) \ge 1-3e^{-\alpha (\epsilon^2/4-\epsilon^3/6) m}, \qquad \forall\, m\in \N.
\end{equation}
Consequently, under the same conditions as in Theorem~\ref{thm:lowerbound}, all the claims in Theorem~\ref{thm:lowerbound} hold with $\beta^{\max}_{\epsilon,\alpha}$ being replaced by $\mathring{\beta}^{\max}_{\epsilon,\alpha}$.
\end{theo}

\begin{proof}
Let $\gamma:=\lfloor \beta m\rfloor/m$ and $k:=\gamma m$. Then $\PP(\mathring{\I}_{\epsilon,\beta})=\PP(\mathring{\I}_{\epsilon,\gamma})$. By $\|y_T\|\le \|y\|$ and $0<\alpha<1$,
it follows from \eqref{yT} that
\begin{align*}
\PP(\mathring{\I}_{\epsilon,\gamma})&=\PP\left\{ 1-\epsilon\le \min_{T\in T_\gamma} \tfrac{1}{m}{\|y_T\|^2}\le \max_{T\in T_\gamma} \tfrac{1}{m}{\|y_T\|^2}\le 1+\epsilon\right\}\\
&\ge \PP\left\{ (1-\epsilon)m \le \|y\|^2-(y_{(1)}^2+\cdots+y_{(k)}^2)\quad \mbox{and}\quad \tfrac{1}{m}{\|y\|^2} \le 1+\sqrt{\alpha} \epsilon\right\}\\
&\ge \PP\left\{y_{(1)}^2+\cdots+y_{(k)}^2\le (1-\sqrt{\alpha})\epsilon m\quad \mbox{and}\quad |\tfrac{1}{m}{\|y\|^2}-1|\le \sqrt{\alpha} \epsilon\right\}\\
&=\PP(E_0\cap E_3)\ge 1-\PP(E_0^c)-\PP(E_3^c),
\end{align*}
where $E_0:=\{|\tfrac{1}{m}{\|y\|^2}-1|\le  \sqrt{\alpha} \epsilon\}$ as in
\eqref{E0} and  $E_3:=\{ y_{(1)}^2+\cdots+y_{(k)}^2\le (1-\sqrt{\alpha})\epsilon m\}$.
By \eqref{grm:normpreservation} and $0<\alpha<1$, we have
\[
\PP(E_0^c)\le 2 e^{-(\alpha \epsilon^2/4-\alpha^{3/2}\epsilon^3/6) m}\le 2 e^{-\alpha(\epsilon^2/4-\epsilon^3/6) m}.
\]
Recall from \eqref{eq:pr2}
that the following inequality holds  for any $\delta>0$:
\begin{equation}\label{eq:2pr2}
\PP\left\{\sqrt{\frac{1}{k} \sum_{j=1}^k y_{(j)}^2 }\le \delta +\sqrt{2\ln \frac{e}{\gamma}}\right\}\ge \PP\left\{\sqrt{\frac{1}{k} \sum_{j=1}^k y_{(j)}^2 }\le \delta +\E \sqrt{\frac{1}{k} \sum_{j=1}^k y_{(j)}^2 }\right\}
\ge 1-e^{-\delta^2 \gamma m/2}.
\end{equation}
Set
\[
\delta:=\sqrt{\tfrac{(1-\sqrt{\alpha})\epsilon}{\gamma}}-\sqrt{2\ln \tfrac{e}{\gamma}}.
\]
Since the function $\beta \ln \frac{e}{\beta}$ is an increasing function on $(0,1]$, by $0<\gamma\le \beta<1$, we deduce from   \eqref{eq:312} that
\[
0<\gamma \ln \frac{e}{\gamma}\le \beta\ln\frac{e}{\beta}\le \frac{\epsilon}{2}\left( \sqrt{1-\sqrt{\alpha}}-\sqrt{\tfrac{\alpha \epsilon}{2}}\right)^2.
\]
Since $\gamma>0$, dividing $\frac{\gamma}{2}$ on both sides and then taking square root on the above inequality, we see that
\[
\sqrt{2\ln \frac{e}{\gamma}} \le \frac{\sqrt{\epsilon}}{\gamma}\left(\sqrt{1-\sqrt{\alpha}}-\sqrt{\tfrac{\alpha \epsilon}{2}}\right)=
\sqrt{\frac{(1-\sqrt{\alpha})\epsilon}{\gamma}}-\sqrt{\frac{\alpha \epsilon^2}{2\gamma}},
\]
from which it is trivial to see that
%\begin{equation}\label{cond:2beta:alpha}
$\delta\ge\sqrt{\frac{\alpha \epsilon^2}{2\gamma}}>0
$
%\end{equation}
%
holds.
By the definition of the set $E_3$,
it follows from \eqref{eq:2pr2} and  $\delta\ge \sqrt{\frac{\alpha \epsilon^2}{2\gamma}}>0$
%\eqref{cond:2beta:alpha}
that
\begin{align*}
\PP(E^c_3)&=\PP\left\{ \sqrt{\frac{1}{k} \sum_{j=1}^k y_{(j)}^2}> \sqrt{\frac{(1-\sqrt{\alpha})\epsilon}{\gamma}}\right\}
=\PP\left\{ \sqrt{\frac{1}{k} \sum_{j=1}^k y_{(j)}^2}> \delta+\sqrt{2\ln \frac{e}{\gamma}} \right\}\\
&\le e^{-\delta^2 \gamma m/2}\le e^{-\alpha \epsilon^2 m/4}
\le e^{-\alpha(\epsilon^2/4-\epsilon^3/6)m}.
\end{align*}
Therefore, $\PP(\mathring{\I}_{\epsilon,\beta})=\PP(\mathring{\I}_{\epsilon,\gamma})
\ge 1-\PP(E_0^c)-\PP(E_3^c)\ge 1-3e^{-\alpha (\epsilon^2/4-\epsilon^3/6)m}$.
This proves \eqref{eq:th1:0}.

It has been proved in the proof of Theorem~\ref{thm:lowerbound} that \eqref{beta:epsalpha},
combined with $0<\alpha<1$ and $0<\epsilon \le \min(1,\frac{1-\sqrt{\alpha}}{4\alpha})$, implies the conditions in \eqref{cond:beta} with $\le$ being replaced by $<$.
Therefore, all the claims in
Theorem~\ref{thm:lowerbound} hold with $\beta^{\max}_{\epsilon,\alpha}$ being replaced by $\mathring{\beta}^{\max}_{\epsilon,\alpha}$.
\end{proof}

To provide an upper bound for $\mathring{\beta}^{\max}_{\epsilon,\alpha}$, we define $\mathring{\beta}^{\max}_{\epsilon,\alpha,C}$ as in \eqref{betamax:2} with $\I_{\epsilon,\beta}$ being replaced by $\mathring{\I}_{\epsilon,\beta}$. Trivially, $\mathring{\beta}^{\max}_{\epsilon,\alpha}\le \mathring{\beta}^{\max}_{\epsilon,\alpha,3}$ for all $\epsilon>0$ and $\alpha>0$.

\begin{theo}\label{thm:upperbound:0}
Let $c_g$ be the absolute positive constant in \eqref{cg}. Then
\begin{equation}\label{eq:upperbound:0}
\mathring{\beta}^{\max}_{\epsilon,\alpha, C} \le \frac{\epsilon}{-c_g^2 W_1(-\epsilon/(2c_g^2))}< \frac{\epsilon}{2c_g^2 \ln\tfrac{2c_g^2}{\epsilon}},
\qquad\forall\, 0<\epsilon<\min(1, c_g^2\ln 2), \alpha>0, C>0.
\end{equation}
\end{theo}

\begin{proof}
If $\mathring{\beta}^{\max}_{\epsilon,\alpha,C}=0$, the claim is trivially true. Hence, we assume $\mathring{\beta}^{\max}_{\epsilon,\alpha,C}>0$. We first prove that
\begin{equation}\label{cond:beta:max:0}
g_\epsilon(\beta):=c_g \sqrt{\beta\ln \tfrac{2}{\beta}}-\sqrt{\epsilon}\le 0 \qquad \forall\; 0<\beta<\min(\tfrac{1}{2}, \mathring{\beta}^{\max}_{\epsilon,\alpha,C}).
\end{equation}
It suffices to prove \eqref{cond:beta:max:0} for rational numbers $\beta$. Suppose not.
Then there exists $\tilde{\epsilon}>0$ such that $\delta:=c_g\sqrt{\beta\ln \frac{2}{\beta}}-\sqrt{\epsilon+\tilde{\epsilon}}>0$.
Let $y=Ax_0$ and $k:=\beta m$. Then
\[
\PP(\mathring{\I}_{\epsilon,\beta})\le \PP\{\min_{T\in T_\beta} \|y_T\|^2\ge (1-\epsilon)m\}
\le \PP\{\tfrac{1}{m}\|y\|^2>1+\tilde{\epsilon}\}+\PP\{y_{(1)}^2+\cdots+y_{(k)}^2\le \eta m\}
\]
with $\eta:=\epsilon+\tilde{\epsilon}$. The same argument as in Lemma~\ref{lem:beta} yields
$\PP(\mathring{\I}_{\epsilon,\beta})\le e^{-\tilde{\epsilon}^2m/4}+e^{-\delta^2 m/2}$.
Since $0<\beta<\mathring{\beta}^{\max}_{\epsilon,\alpha,C}$, the definition of $\mathring{\beta}^{\max}_{\epsilon,\alpha,C}$ implies $\PP(\mathring{\I}_{\epsilon,\beta})\ge 1-3e^{-\alpha (\epsilon^2/4-\epsilon^3/6) m}$ and
the same argument as in Theorem~\ref{thm:upperbound} leads to a contradiction. Therefore, \eqref{cond:beta:max:0} must hold.

Note that $g_\epsilon$ is an increasing function on $(0,\frac{2}{e})$ and $\frac{1}{2}<\frac{2}{e}$. By $\frac{\epsilon}{c_g^2}\le \ln 2$ and the simple fact $-W_{-1}(x) <-2\ln(-x)$ for all $x\in (-e^{-1},0)$,
it is easy to conclude from \eqref{cond:beta:max:0} that \eqref{eq:upperbound:0} must hold.
\end{proof}

\subsection{Proof of Theorem~\ref{thm:optimal}}

We are ready to prove Theorem~\ref{thm:optimal}.

\begin{proof}[Proof of Theorem~\ref{thm:optimal}]
The left-hand inequality in \eqref{eq:optimal} follows directly from \eqref{beta:eps} of Theorem~\ref{thm:lowerbound}.
Since $0<\epsilon<4\epsilon_g^2$, we have $\sqrt{\epsilon}\le 2\epsilon_g$ and therefore, $\ln\tfrac{2\epsilon_g}{\epsilon}\ge \ln \tfrac{\sqrt{\epsilon}}{\epsilon}=\tfrac{1}{2}\ln \tfrac{1}{\epsilon}$. By Theorem~\ref{thm:upperbound}, we have
\[
\beta^{\max}_{\epsilon,\alpha}\le \beta^{\max}_{\epsilon,\alpha,3}< \frac{(1+\epsilon_g^{-1})\epsilon}{c_g^2 \ln \tfrac{2\epsilon_g}{\epsilon}}\le
\frac{(1+\epsilon_g^{-1})\epsilon}{\tfrac{1}{2}c_g^2 \ln \tfrac{1}{\epsilon}}
=\frac{2+2\epsilon_g}{c_g^2\epsilon_g} \frac{\epsilon}{\ln \tfrac{1}{\epsilon}}.
\]
This proves the right-hand inequality in \eqref{eq:optimal}.

The left-hand inequality in \eqref{eq:optimal:0} follows directly from Theorem~\ref{thm:lowerbound:0} and \eqref{beta:eps} of Theorem~\ref{thm:lowerbound}. Since $\epsilon\le \frac{1}{2c_g^2}$, we have $\ln\tfrac{1}{\epsilon}\ge \ln (2c_g^2)$. Note that $0<\min(c_g^2\ln2, \tfrac{1}{2c_g^2})<1$.
Now the right-hand inequality in \eqref{eq:optimal:0} follows directly from Theorem~\ref{thm:upperbound:0}.
\end{proof}

%%%%%%%%%%%%%%%%%%%%%%%%%%%%%%%%%%%%%%%%%%%%%%%%%%%%%%%%%
\section{Gaussian Random Matrices under Arbitrary Erasure of Rows for given $0<\beta<1$}
\label{sec:beta}

In this section we study the robustness property of Gaussian random matrices with arbitrarily erased rows for a given corruption/erasure ratio $0<\beta<1$ with presenting the proof of Theorem~\ref{thm:non}.

\subsection{Estimate $\lcst^{\max}_{\beta}(\alpha)$ and $\ucst^{\min}_{\beta}(\alpha)$}

\begin{comment}
Let us first define some quantities related to $\lcst^{\max}_{\beta}(\alpha)$ and $\ucst^{\min}_{\beta}(\alpha)$ in \eqref{lcst:alpha} and \eqref{ucst:alpha} by relaxing the requirement on all $m\in \N$.
Set
%
\begin{align}
&\lcst^{\max}_{\beta,\infty}(\alpha):=\sup\{0\le \lcst\le \infty \; : \;  \PP(\I_{[\lcst,\infty],\beta}) \ge 1-\exp(-\alpha m)\quad \mbox{for sufficiently large}\; m\in \N\},\label{lcst:alpha:infty}\\
&\ucst^{\min}_{\beta,\infty}(\alpha):=\inf\{0\le \ucst\le \infty \; : \; \PP(\I_{[0,\ucst],\beta}) \ge 1-\exp(-\alpha m)\quad \mbox{for sufficiently large}\; m\in \N\}, \label{ucst:alpha:infty}
\end{align}
%
and
%
\begin{equation}\label{cst:infty}
\lcst^{\max}_{\beta,\infty}:=\sup\{\lcst^{\max}_{\beta,\infty}(\alpha) \; : \; \alpha>0\},\qquad
\ucst^{\min}_{\beta,\infty}:=\inf\{\ucst^{\min}_{\beta,\infty}(\alpha) \; : \; \alpha>0\}.
\end{equation}
%
Note that $\lcst^{\max}_{\beta,\infty}(\alpha)$ is a decreasing function of $\alpha$ and $\ucst^{\min}_{\beta,\infty}(\alpha)$ is an increasing function of $\alpha$. Hence,
it is trivial to observe that
\[
\lcst^{\max}_{\beta,\infty}=\lim_{\alpha\to 0^+} \lcst^{\max}_{\beta}(\alpha), \quad
\ucst^{\min}_{\beta,\infty}=\lim_{\alpha\to 0^+} \ucst^{\min}_{\beta}(\alpha)
\]
and
%
\begin{equation}\label{lcst:ucst:infty}
\lcst^{\max}_{\beta}(\alpha)\le
\lcst^{\max}_{\beta,\infty}(\alpha) \le \lcst^{\max}_{\beta,\infty}\le
\ucst^{\min}_{\beta,\infty}\le
\ucst^{\min}_{\beta,\infty}(\alpha)
\le \ucst^{\min}_{\beta}(\alpha).
\end{equation}
%
\end{comment}

To prove Theorem \ref{thm:non}, we first estimate $\lcst^{\max}_{\beta}(\alpha)$.
%and $\lcst^{\max}_{\beta,\infty}(\alpha)$.

\begin{lem}\label{lem:alpha:lcst}
For $0<\beta<1$ and $0<\alpha< \frac{\pi}{12}(1-\beta)^2h_\beta$ with $h_\beta:=\min(\tfrac{3}{4}-\tfrac{1}{2}\beta, 1-\beta)$,
\begin{equation}\label{est:lcst}
0<\frac{\pi}{6}(1-\beta)h_\beta
+\frac{2\alpha}{1-\beta}-2
\sqrt{\pi \alpha h_\beta/3}
\le \lcst^{\max}_{\beta}(\alpha) \le %\lcst^{\max}_{\beta,\infty}(\alpha) \le
\min\Big(\frac{\pi}{2}\left(\ln \frac{1}{\beta}\right)^2,1\Big).
\end{equation}
\end{lem}

\begin{proof}
Define $\gamma:= {\lfloor\beta m \rfloor}/{m}$ and $k:=\gamma m$.
Let $y:=(y_1,\ldots, y_m)^\tp:=A x_0$. Since $\gamma m=\lfloor \beta m\rfloor$,
by Lemma~\ref{lem:dec} and \eqref{yT}, for $\lcst\ge 0$, we have
\begin{equation}\label{P:lcst}
\PP(\I_{[\lcst,\infty],\beta})=\PP\left\{ \min_{T\in T_\gamma} \tfrac{1}{|T|}\|y_T\|^2\ge \lcst\right\}=
%\PP\left\{\frac{1}{m-k}\sum_{j=k+1}^m y_{(j)}^2 \ge \lcst \right\}=
\PP\left\{\sqrt{\frac{1}{m-k}\sum_{j=k+1}^m y_{(j)}^2} \ge \sqrt{\lcst} \right\}.
\end{equation}
By the left-hand inequality in \eqref{est:small} and $0\le \gamma\le \beta<1$, we have
\begin{align*}
\E\sqrt{\frac{1}{m-k}\sum_{j=k+1}^m y_{(j)}^2}
&\ge \sqrt{\frac{\pi}{6}} \sqrt{(1-\gamma)\frac{1-\gamma+\frac{1}{2m}}{1+\frac{1}{m}}}
\ge \sqrt{\frac{\pi}{6}} \sqrt{(1-\beta)\frac{1-\beta+\frac{1}{2m}}{1+\frac{1}{m}}}\\
&\ge \sqrt{\frac{\pi}{6}} \sqrt{(1-\beta)}\sqrt{\inf_{0<x\le 1} \frac{1-\beta+\frac{1}{2}x}{1+x}}
=\sqrt{\frac{\pi}{6}} \sqrt{(1-\beta)h_\beta}.
\end{align*}
Therefore, for $0<\beta<1$, by \eqref{P:lcst} and \eqref{est:mean:2} of Lemma~\ref{lem:Lip} with $\delta:=\sqrt{\frac{\pi}{6}}\sqrt{(1-\beta)h_\beta}-\sqrt{\lcst}$, if $\delta\ge \sqrt{\frac{2\alpha}{1-\beta}}>0$, then we have
\begin{align*}
\PP(\I_{[\lcst,\infty],\beta})
&=
\PP\left\{\sqrt{\frac{1}{m-k}\sum_{j=k+1}^m y_{(j)}^2}-\E\sqrt{\frac{1}{m-k}\sum_{j=k+1}^m y_{(j)}^2} \ge \sqrt{\lcst}-\E\sqrt{\frac{1}{m-k}\sum_{j=k+1}^m y_{(j)}^2} \right\}\\
&\ge \PP\left\{\sqrt{\frac{1}{m-k}\sum_{j=k+1}^m y_{(j)}^2}-\E\sqrt{\frac{1}{m-k}\sum_{j=k+1}^m y_{(j)}^2} \ge \sqrt{\lcst}-\sqrt{\frac{\pi}{6}}\sqrt{(1-\beta)h_\beta} \right\}\\
%&=\PP\left\{\sqrt{\frac{1}{m-k}\sum_{j=k+1}^m y_{(j)}^2}-\E\sqrt{\frac{1}{m-k}\sum_{j=k+1}^m y_{(j)}^2} \ge -\delta \right\}\\
&\ge 1-e^{-\delta^2(m-k)/2}\ge 1-e^{-\alpha \frac{m-k}{1-\beta}}=
1-e^{-\frac{1-\gamma}{1-\beta} \alpha m}\ge 1-e^{-\alpha m},
\end{align*}
since $\frac{1-\gamma}{1-\beta}\ge 1$ by $0\le \gamma\le \beta<1$.
This shows that
if $\sqrt{\lcst}\le \sqrt{\frac{\pi}{6}}\sqrt{(1-\beta)h_\beta}-\sqrt{\frac{2\alpha}{1-\beta}}$, then
\begin{equation} \label{P:lcst:all}
\PP(\I_{[\lcst,\infty],\beta})\ge 1-e^{-\alpha m}
\end{equation}
for all $m\in \N$.
Since $0<\alpha<\frac{\pi}{12}(1-\beta)^2 h_\beta$, we have $\sqrt{\frac{\pi}{6}}\sqrt{(1-\beta)h_\beta}-\sqrt{\frac{2\alpha}{1-\beta}}> 0$. Consequently, by the definition of $\lcst^{\max}_{\beta}(\alpha)$, we conclude that
\[
\lcst^{\max}_{\beta}(\alpha) \ge \left(\sqrt{\frac{\pi}{6}}\sqrt{(1-\beta)h_\beta}-\sqrt{\frac{2\alpha}{1-\beta}}\right)^2
=\frac{\pi}{6}(1-\beta)h_\beta+\frac{2\alpha}{1-\beta}
-2\sqrt{\pi\alpha h_\beta/3}>0.
\]
This proves the left-hand side of \eqref{est:lcst}.

We now estimate the upper bound for $\lcst^{\max}_{\beta}(\alpha)$.
By the second inequality in \eqref{gmin:bound} with $p=1$, we have
\begin{equation}\label{lcst:est:E:upper}
\E\sqrt{\frac{1}{m-k}\sum_{j=k+1}^m y_{(j)}^2}\le
\E|y_{(k+1)}|\le \sqrt{\frac{\pi}{2}} \left(\frac{1}{k+1}+\ln\frac{m}{k+1}\right)\le \sqrt{\frac{\pi}{2}}\ln \frac{m}{k}=\sqrt{\frac{\pi}{2}} \ln \frac{1}{\gamma},
\end{equation}
where we used the basic inequality
$\frac{1}{k+1}\le \ln \left(1+\frac{1}{k}\right)$ for all $k>0$.
Suppose that \eqref{P:lcst:all} holds for sufficiently large $m\in \N$. For convenience, we only consider the case
that $\beta$ is rational, since the general result follows from the fact that the rational numbers are dense in $\R$. We assume that $m\in \N$ is sufficiently large and satisfies $m\beta\in \N$, that is, we have $\gamma=\beta$. Note that $k=\gamma m=\beta m$.
By \eqref{P:lcst} and \eqref{lcst:est:E:upper}, applying \eqref{est:mean:1} of Lemma~\ref{lem:Lip} with $\delta:=\sqrt{\lcst}-\sqrt{\frac{\pi}{2}}\ln \frac{1}{\beta}>0$, we have
\begin{align*}
\PP(\I_{[\lcst,\infty],\beta})
&=
\PP\left\{\sqrt{\frac{1}{m-k}\sum_{j=k+1}^m y_{(j)}^2}-\E\sqrt{\frac{1}{m-k}\sum_{j=k+1}^m y_{(j)}^2}\ge \sqrt{\lcst}-\E\sqrt{\frac{1}{m-k}\sum_{j=k+1}^m y_{(j)}^2} \right\}\\
&\le \PP\left\{\sqrt{\frac{1}{m-k}\sum_{j=k+1}^m y_{(j)}^2}-\E\sqrt{\frac{1}{m-k}\sum_{j=k+1}^m y_{(j)}^2}\ge \sqrt{\lcst}-\sqrt{\frac{\pi}{2}} \ln \frac{1}{\beta}\right\}\\
&\le e^{-\delta^2(m-k)/2}=e^{-\delta^2(1-\beta) m/2}.
\end{align*}
Consequently, if $\sqrt{\lcst}>\sqrt{\frac{\pi}{2}}\ln \frac{1}{\beta}$ and if \eqref{P:lcst:all} holds for sufficiently large $m\in \N$, then the above inequalities imply
\[
1-e^{-\alpha m}\le \PP(\I_{[\lcst,\infty],\beta}) \le e^{-\delta^2 (1-\beta) m/2},
\]
which cannot be true for sufficiently large $m$ since $\alpha>0$ and $\delta>0$. This proves that
$\lcst^{\max}_\beta(\alpha)\le
%\lcst^{\max}_{\beta,\infty}(\alpha)\le
\frac{\pi}{2}(\ln \frac{1}{\beta})^2$.
Also, it is trivial to see that
\[
\E\sqrt{\frac{1}{m-k}\sum_{j=k+1}^m y_{(j)}^2}\le
\sqrt{\frac{1}{m-k}\sum_{j=k+1}^m \E y_{(j)}^2}\le \sqrt{\frac{1}{m}\sum_{j=1}^m \E y_{(j)}^2}=1.
\]
The above same argument shows that $\lcst^{\max}_\beta(\alpha)
%\le \lcst^{\max}_{\beta,\infty}(\alpha)
\le 1$. This proves the upper bound of \eqref{est:lcst}.
\end{proof}

We next estimate $\ucst^{\min}_{\beta}(\alpha)$.
%and $\ucst^{\min}_{\beta,\infty}(\alpha)$.

\begin{lem}\label{lem:alpha:ucst}
For $0<\beta<1$ and $\alpha>0$,
\begin{equation}\label{est:ucst}
\max\Big(c_g^2 \ln \frac{2}{1-\beta},\frac{\pi}{2}\beta^2\Big)\le
%\ucst^{\min}_{\beta,\infty}(\alpha)\le
\ucst^{\min}_{\beta}(\alpha)
\le 2\ln \frac{e}{1-\beta}+\frac{2\alpha}{1-\beta}+4\sqrt{\frac{\alpha}{1-\beta}\ln \frac{e}{1-\beta}}.
\end{equation}
\end{lem}

\begin{proof}
Define $\gamma:= {\lfloor\beta m \rfloor}/{m}$ and $k:=\gamma m$.
Let $y:=(y_1,\ldots, y_m)^\tp:=A x_0$. Since $\gamma m=\lfloor \beta m\rfloor$,
by Lemma~\ref{lem:dec} and \eqref{yT}, for $\ucst\ge 0$, we have
\begin{equation}\label{P:ucst}
\PP(\I_{[0,\ucst],\beta})=\PP\left\{\max_{T\in T_\gamma} \tfrac{1}{|T|}\|y_T\|^2\le \ucst\right\}
%=\PP\left\{ \frac{1}{m-k} \sum_{j=1}^{m-k} y_{(j)}^2\le \ucst\right\}
=\PP\left\{ \sqrt{\frac{1}{m-k} \sum_{j=1}^{m-k} y_{(j)}^2}\le \sqrt{\ucst}\right\}.
\end{equation}
By \eqref{le:upperbound:new}, we have
\[
\E \sqrt{\frac{1}{m-k} \sum_{j=1}^{m-k} y_{(j)}^2}
\le \sqrt{2\ln \frac{em}{m-k}}
=\sqrt{2\ln \frac{e}{1-\gamma}}.
%\le  \sqrt{\ln \frac{e}{1-\beta}}.
\]
By \eqref{P:ucst} and the above estimate,
applying \eqref{est:mean:1} of Lemma~\ref{lem:Lip} with $\delta:=\sqrt{\ucst}-\sqrt{2 \ln \frac{e}{1-\gamma}}\ge \sqrt{\frac{2\alpha}{1-\gamma}}>0$, we have
\begin{equation}\label{est:needed}
\begin{split}
\PP(\I_{[0,\ucst],\beta})
&=\PP\left\{ \sqrt{\frac{1}{m-k} \sum_{j=1}^{m-k} y_{(j)}^2}-\E \sqrt{\frac{1}{m-k} \sum_{j=1}^{m-k} y_{(j)}^2}\le \sqrt{\ucst}-\E \sqrt{\frac{1}{m-k} \sum_{j=1}^{m-k} y_{(j)}^2}\right\}\\
&\ge \PP\left\{ \sqrt{\frac{1}{m-k} \sum_{j=1}^{m-k} y_{(j)}^2}-\E \sqrt{\frac{1}{m-k} \sum_{j=1}^{m-k} y_{(j)}^2}\le \sqrt{\ucst}- \sqrt{2\ln \frac{e}{1-\gamma}} \right\}\\
&\ge 1-e^{-\delta^2 (m-k)/2}\ge 1-e^{-\alpha \frac{m-k}{1-\gamma}}=1-e^{-\alpha m}.
\end{split}
\end{equation}
If $\sqrt{\ucst}\ge \sqrt{\frac{2\alpha}{1-\beta}}+\sqrt{2\ln \frac{e}{1-\beta}}$, then we have $\sqrt{\ucst}\ge \sqrt{\frac{2\alpha}{1-\beta}}+\sqrt{2\ln \frac{e}{1-\beta}} \ge \sqrt{\frac{2\alpha}{1-\gamma}}+\sqrt{2\ln \frac{e}{1-\gamma}}$ by $0\le \gamma\le \beta<1$ and the above inequality shows that
\begin{equation}\label{P:ucst:all}
\PP(\I_{[0,\ucst],\beta})\ge 1-e^{-\alpha m}
\end{equation}
holds for all $m\in \N$.
Therefore, we proved
\[
\ucst^{\min}_{\beta}(\alpha)\le \left(\sqrt{\frac{2\alpha}{1-\beta}}+\sqrt{2\ln \frac{e}{1-\beta}}\right)^2=2\ln \frac{e}{1-\beta}+\frac{2\alpha}{1-\beta}+4\sqrt{\frac{\alpha}{1-\beta}\ln \frac{e}{1-\beta}}.
\]
This proves the right-hand side inequality in \eqref{est:ucst}.

Without loss of generality, we assume that $\beta$ is a rational number and $m$ is sufficiently large satisfying $\beta m\in \N$. Thus, $\gamma=\beta$ and $k=\beta m$. We consider two cases of $\beta$. Suppose that $1/2\le \beta<1$. Then $m-k=(1-\beta)m\le m/2$ and by \eqref{cg}, we have
\[
\E \sqrt{\frac{1}{m-k} \sum_{j=1}^{m-k} y_{(j)}^2}\ge \E |y_{(m-k)}|\ge c_g \sqrt{\ln \frac{2m}{m-k}}=c_g \sqrt{\ln \frac{2}{1-\beta}}.
\]
By \eqref{P:ucst} and the above inequality, applying \eqref{est:mean:2} of Lemma~\ref{lem:Lip} with $\delta:=c_g \sqrt{\ln \frac{2}{1-\beta}}-\sqrt{\ucst}>0$ and $S=\{1,\ldots,m-k\}$, we have
\begin{align*}
\PP(\I_{[0,\ucst],\beta})
&=\PP\left\{ \sqrt{\frac{1}{m-k} \sum_{j=1}^{m-k} y_{(j)}^2}-\E \sqrt{\frac{1}{m-k} \sum_{j=1}^{m-k} y_{(j)}^2}\le \sqrt{\ucst}-\E \sqrt{\frac{1}{m-k} \sum_{j=1}^{m-k} y_{(j)}^2}\right\}\\
&\le \PP\left\{ \sqrt{\frac{1}{m-k} \sum_{j=1}^{m-k} y_{(j)}^2}-\E \sqrt{\frac{1}{m-k} \sum_{j=1}^{m-k} y_{(j)}^2}\le \sqrt{\ucst}-c_g \sqrt{\ln \frac{2}{1-\beta}} \right\}\\
&=\PP\left\{ \sqrt{\frac{1}{m-k} \sum_{j=1}^{m-k} y_{(j)}^2}-\E \sqrt{\frac{1}{m-k} \sum_{j=1}^{m-k} y_{(j)}^2}\le -\delta \right\}\\
&\le e^{-\delta^2 (m-k)/2}=e^{-\delta^2(1-\beta) m/2}.
\end{align*}
Consequently, if \eqref{P:ucst:all} holds for sufficiently large $m\in \N$, then
\[
1-e^{-\alpha m}\le \PP(\I_{[0,\ucst],\beta})\le e^{-\delta^2(1-\beta) m/2},
\]
which cannot be true when $m$ is sufficiently large. This proves that
$\ucst^{\min}_{\beta}(\alpha)\ge c_g^2 \ln \frac{2}{1-\beta}$ provided that $1/2\le \beta<1$.

Suppose that $0<\beta<1/2$. By \eqref{gmin:bound}, we have
\[
\E \sqrt{\frac{1}{m-k} \sum_{j=1}^{m-k} y_{(j)}^2}\ge \E |y_{(m-k)}|\ge \sqrt{\frac{\pi}{2}} \frac{k+1}{m+1}\ge \sqrt{\frac{\pi}{2}} \frac{k}{m}=\sqrt{\frac{\pi}{2}} \beta.
\]
The above same argument shows that $\ucst^{\min}_{\beta}(\alpha)\ge \frac{\pi}{2}\beta^2$ provided that $0<\beta<1/2$.
This proves the left-hand side inequality in \eqref{est:ucst}.
\end{proof}

\subsection{Estimate $\mathring{\lcst}^{\max}_{\beta}(\alpha)$ and $\mathring{\ucst}^{\min}_{\beta}(\alpha)$}

\begin{comment}
Let $x_0\in \R^\Ndim$ with $\|x_0\|=1$.
For $0\le \beta<1$ and $0\le \lcst\le \ucst\le \infty$,
recall that $\mathring{\I}_{[\lcst,\ucst],\beta}$ is defined in \eqref{I:interval:2}.
We can define $\mathring{\lcst}^{\max}_{\beta}(\alpha), \mathring{\ucst}^{\min}_{\beta}(\alpha), \mathring{\lcst}^{\max}_{\beta}, \mathring{\ucst}^{\min}_{\beta}$
and $\mathring{\lcst}^{\max}_{\beta,\infty}(\alpha), \mathring{\ucst}^{\min}_{\beta,\infty}(\alpha), \mathring{\lcst}^{\max}_{\beta,\infty}, \mathring{\ucst}^{\min}_{\beta,\infty}$, which
are similar to $\lcst^{\max}_{\beta}(\alpha)$, $\ucst^{\min}_{\beta}(\alpha)$, $\lcst^{\max}_{\beta}$, $\ucst^{\min}_{\beta}$ and  $\lcst^{\max}_{\beta,\infty}(\alpha)$, $\ucst^{\min}_{\beta,\infty}(\alpha)$, $\lcst^{\max}_{\beta,\infty}$, $\ucst^{\min}_{\beta,\infty}$, respectively by replacing $\I$ with $\mathring{\I}$.
\end{comment}

As a direct consequence of Lemma~\ref{lem:alpha:lcst}, we have

\begin{coro}\label{cor:alpha:mlcst}
For $0<\beta<1$ and $0<\alpha< \frac{\pi}{12}(1-\beta)^2h_\beta$ with $h_\beta:=\min(\tfrac{3}{4}-\tfrac{1}{2}\beta, 1-\beta)$,
\begin{equation}\label{est:mlcst}
0<\frac{\pi}{6}(1-\beta)^2h_\beta
+2\alpha-2(1-\beta)
\sqrt{\pi \alpha h_\beta/3}
\le \mathring{\lcst}^{\max}_{\beta}(\alpha) \le %\mathring{\lcst}^{\max}_{\beta,\infty}(\alpha) \le
(1-\beta)\min\Big(\frac{\pi}{2}\left(\ln \frac{1}{\beta}\right)^2,1\Big).
\end{equation}
\end{coro}

\begin{proof}
Define $\gamma:= {\lfloor\beta m \rfloor}/{m}$ and $k:=\gamma m$. Let $y:=(y_1,\ldots, y_m)^\tp:=A x_0$. Then $0\le \gamma\le \beta<1$.
By the definition of $\mathring{\I}_{[\lcst,\infty],\beta}$, we have
\[
\PP(\mathring{\I}_{[\lcst,\infty],\beta})
=\PP\left\{\sqrt{\frac{1}{m}\sum_{j=k+1}^m y_{(j)}^2} \ge \sqrt{\lcst} \right\}
=\PP\left\{\sqrt{\frac{1}{m-k}\sum_{j=k+1}^m y_{(j)}^2} \ge \sqrt{\frac{\lcst}{1-\gamma}} \right\}=\PP(\I_{[\frac{\lcst}{1-\gamma},\infty],\beta})
\ge \PP(\I_{[\frac{\lcst}{1-\beta},\infty],\beta}),
\]
where we used $\frac{1}{1-\gamma}\le \frac{1}{1-\beta}$ by $0\le \gamma\le \beta<1$. Consequently, for all $0\le \lcst< (1-\beta)\lcst^{\max}_{\beta}(\alpha)$, we have $\frac{\lcst}{1-\beta}<\lcst^{\max}_{\beta}(\alpha)$ and by the definition of $\lcst^{\max}_{\beta}(\alpha)$, we have
\[
1-e^{-\alpha m}\le \PP(\I_{[\frac{\lcst}{1-\beta},\infty],\beta}) \le \PP(\mathring{\I}_{[\lcst,\infty],\beta}),\qquad \forall\; m\in \N.
\]
For $0<\beta<1$ and $0<\alpha\le \frac{\pi}{12}(1-\beta)^2 h_\beta$, it follows from the above inequality and Lemma~\ref{lem:alpha:lcst} that
\[
\mathring{\lcst}^{\max}_{\beta}(\alpha)\ge (1-\beta)\lcst^{\max}_{\beta}(\alpha) \ge \frac{\pi}{6}(1-\beta)^2h_\beta+2\alpha-2(1-\beta)
\sqrt{\pi \alpha h_\beta/3}>0.
\]
This proves the left-hand side inequality in \eqref{est:mlcst}.

Note that we proved the upper bound of $\lcst^{\max}_{\beta}(\alpha)$ in \eqref{est:lcst} of Lemma~\ref{lem:alpha:lcst} by assuming that $\beta$ is rational and $m$ is sufficiently large satisfying $\beta m\in \N$. For such $\beta$ and $m$, we have $\gamma=\beta$ and consequently, the same proof of Lemma~\ref{lem:alpha:lcst} yields
\[
\mathring{\lcst}^{\max}_{\beta}(\alpha)=(1-\beta)
\lcst^{\max}_{\beta}(\alpha)
\le (1-\beta)\min(\frac{\pi}{2}(\ln \tfrac{1}{\beta})^2,1).
\]
This proves the right-hand side inequality in \eqref{est:mlcst}.
\end{proof}

With the help of Lemma~\ref{lem:alpha:ucst}, we have the following result.

\begin{coro}\label{cor:alpha:mucst}
For $0<\beta<1$ and $\alpha>0$,
\begin{equation}\label{est:mucst}
(1-\beta)\max\Big(c_g^2 \ln \frac{2}{1-\beta},\frac{\pi}{2}\beta^2\Big)\le
\mathring{\ucst}^{\min}_{\beta}(\alpha)
%\le2(1-\beta)\ln \frac{e}{1-\beta}+2\alpha+4\sqrt{\alpha(1-\beta)\ln \frac{e}{1-\beta}}
\end{equation}
and
\begin{equation}\label{est:mucst:upper}
\PP(\mathring{\I}_{[0,\ucst],\beta})
\ge 1-e^{-\alpha m},\qquad \forall\; \ucst\ge \left(\sqrt{2(1-\gamma)\ln \frac{e}{1-\gamma}}+\sqrt{2\alpha}\right)^2, m\in \N \quad \mbox{with}\quad
\gamma:=\lfloor \beta m \rfloor/m.
\end{equation}
\end{coro}

\begin{proof}
Define $\gamma:= {\lfloor\beta m \rfloor}/{m}$ and $k:=\gamma m$.
Let $y:=(y_1,\ldots, y_m)^\tp:=A x_0$.
By the definition of $\mathring{\I}_{[0,\ucst],\beta}$,
\[
\PP(\mathring{\I}_{[0,\ucst],\beta})
=\PP\left\{\sqrt{\frac{1}{m}\sum_{j=1}^{m-k} y_{(j)}^2} \le \sqrt{\ucst} \right\}
=\PP\left\{\sqrt{\frac{1}{m-k}\sum_{j=1}^{m-k} y_{(j)}^2} \le \sqrt{\frac{\ucst}{1-\gamma}} \right\}=\PP(\I_{[0,\frac{\ucst}{1-\gamma}],\beta}).
\]
Note that we proved the lower bound of $\ucst^{\min}_{\beta}(\alpha)$ in \eqref{est:ucst} of Lemma~\ref{lem:alpha:ucst} by assuming that $\beta$ is rational and $m$ is sufficiently large satisfying $\beta m\in \N$. For such $\beta$ and $m$, we have $\gamma=\beta$ and the left-hand side inequality in \eqref{est:mucst} follows directly from the same proof of Lemma~\ref{lem:alpha:ucst}.

As proved in \eqref{est:needed}, if $\delta:=\sqrt{\frac{\ucst}{1-\gamma}}-\sqrt{2\ln \frac{e}{1-\gamma}}\ge \sqrt{\frac{2\alpha}{1-\gamma}}>0$, then
\begin{equation}\label{est:mucst:0}
\PP(\mathring{\I}_{[0,\ucst],\beta})
=\PP(\I_{[0,\frac{\ucst}{1-\gamma}],\beta})\ge 1-e^{-\delta^2(m-k)/2}\ge 1-e^{-\alpha m},\qquad \forall\; m\in \N \quad \mbox{with}\quad
\gamma:=\lfloor \beta m \rfloor/m.
\end{equation}
This proves the inequality in \eqref{est:mucst:upper}.
\begin{comment}
Hence, if $\sqrt{w}>\sqrt{2(1-\beta)\ln \frac{e}{1-\beta}}+\sqrt{2\alpha}$, by $\beta-\frac{1}{m}\le \gamma\le \beta$, for sufficiently large $m\in \N$, we have
$\sqrt{w}\ge \sqrt{2(1-\gamma)\ln \frac{e}{1-\gamma}}+\sqrt{2\alpha}$ and consequently,
$\PP(\mathring{\I}_{[0,\ucst],\beta})\ge 1-e^{-\alpha m}$ for sufficiently large $m\in \N$. This proves the right-hand side inequality of \eqref{est:mucst}.
\end{comment}
\end{proof}

To estimate the upper bound of $\mathring{\ucst}^{\min}_{\beta}$, we need the following lemma.

\begin{lem}\label{lem:small:m}
Let $x_0\in \R^n$ with $\|x_0\|=1$ and $m\in \N$.
Let $A$ be an $m\times n$ Gaussian random matrix with i.i.d. entries obeying $\mathcal{N}(0,1)$.  For $0<\beta<1$ and $\alpha>0$, let $\ucst_{\beta,m}(\alpha)$ be the smallest $\ucst\ge 0$ such that
\[
\PP(\mathring{\I}_{[0,\ucst],\beta})=\PP\left\{
\sup_{|T^c|\le \beta m} \frac{1}{m}\|A_T x_0\|^2\le \ucst\right\}\ge 1-e^{-\alpha m}.
\]
Then $\ucst_{\beta,m}(\alpha)>0$ for all $\alpha>0$ and
%
%\begin{equation}\label{ucst:m:to:0}
$\ucst_{\beta,m}:=\lim_{\alpha\to 0^+} \ucst_{\beta,m}(\alpha)=0$.
%\end{equation}
%
\end{lem}

\begin{proof}
Suppose that $\ucst_{\beta,m}(\alpha)=0$. Then $\PP(\mathring{\I}_{[0,0],\beta})\ge 1-e^{-\alpha m}$, which is a contradiction to $\PP(\mathring{\I}_{[0,0],\beta})=0$ and $\alpha>0$. Therefore, we must have $\ucst_{\beta,m}(\alpha)>0$.
Note that $\ucst_{\beta,m}(\alpha)$ is an increasing function of $\alpha$. By $\PP(\mathring{\I}_{[0,\ucst_{\beta,m}(\alpha)],\beta})\ge 1-e^{-\alpha m}$, we have $\PP(\mathring{\I}_{[0,\ucst_{\beta,m}(\alpha)],\beta}^c)\le e^{-\alpha m}$. Then
\[
\PP(\mathring{\I}_{[0,\ucst_{\beta,m}],\beta}^c)
=\lim_{\alpha\to 0^+} \PP(\mathring{\I}_{[0,\ucst_{\beta,m}(\alpha)],\beta}^c) \le \lim_{\alpha\to 0^+} e^{-\alpha m}=1.
\]
Define $y=(y_1,\ldots,y_m)^\tp:=Ax_0$. By $0<\beta<1$, we see that $T\subseteq \{1,\ldots, m\}$ with $|T^c|\le \beta m$ implies $|T|\ge 1$. If $\ucst_{\beta,m}>0$, it is trivial to see that
\[
\PP(\mathring{\I}_{[0,\ucst_{\beta,m}],\beta}^c)\le \PP\left \{ |y_1|^2+\cdots+|y_m|^2>m\ucst_{\beta,m}\right\}<1,
\]
%where we used $\ucst_{\beta,m}>0$ in the last inequality.
which is a contradiction to $\PP(\mathring{\I}_{[0,\ucst_{\beta,m}],\beta}^c)=1$.
This proves $\ucst_{\beta,m}=\lim_{\alpha\to 0^+}\ucst_{\beta,m}(\alpha)=0$.
\end{proof}

\subsection{Proof of Theorem~\ref{thm:non}}

We are now ready to prove Theorem~\ref{thm:non}.

\begin{proof}[Proof of Theorem~\ref{thm:non}]
Observe that $\lcst^{\max}_\beta
=\lim_{\alpha\to 0^+} \lcst^{\max}_{\beta}(\alpha)$
and
$\ucst^{\min}_\beta=\lim_{\alpha \to 0^+} \ucst^{\min}_{\beta}(\alpha)$.
Taking $\alpha\to 0^+$ in \eqref{est:lcst} of
Lemma~\ref{lem:alpha:lcst}, we have
\begin{equation}
\frac{\pi}{6}(1-\beta)^2 \min\Big( \frac{3-2\beta}{4(1-\beta)},1\Big) \le \lcst_\beta^{\max} \le
%\lcst^{\max}_{\beta,\infty} \le
\min\left(\frac{\pi}{2} \Big(\ln \frac{1}{\beta}\Big)^2,1\right).
\end{equation}
This proves \eqref{eq:theta}.
Taking $\alpha\to 0^+$ in \eqref{est:ucst} of Lemma~\ref{lem:alpha:ucst}, we have
\begin{equation}
\max\Big(c_g^2 \ln \frac{2}{1-\beta}, \frac{\pi}{2}\beta^2\Big)
%\le \ucst_{\beta,\infty}^{\min}
\le \ucst_\beta^{\min} \le 2\ln \frac{e}{1-\beta}.
\end{equation}
This proves \eqref{eq:omega}.
Similarly, taking $\alpha\to 0^+$ in \eqref{est:mlcst} of Corollary~\ref{cor:alpha:mlcst}, we have
\begin{equation}
\frac{\pi}{6}(1-\beta)^3 \min\Big( \frac{3-2\beta}{4(1-\beta)},1\Big) \le \mathring{\lcst}_\beta^{\max} \le %\mathring{\lcst}^{\max}_{\beta,\infty} \le
(1-\beta) \min\left(\frac{\pi}{2} \Big(\ln \frac{1}{\beta}\Big)^2,1\right).
\end{equation}
This proves \eqref{eq:theta:2}. We now prove \eqref{eq:omega:2}.
Taking $\alpha\to 0^+$ for the left-hand side of \eqref{est:mucst} in Corollary~\ref{cor:alpha:mucst}, we proved the left-hand side of \eqref{eq:omega:2}. We now prove the right-hand side of \eqref{eq:omega:2}. Take $N\in \N$ such that $N>\frac{1}{\beta}$. Let $\gamma:=\lfloor \beta m\rfloor/m$. Then $\beta-\frac{1}{N}\le \beta-\frac{1}{m}\le \gamma\le \beta$ for all $m\ge N$. Since $(1-x)\ln \frac{e}{1-x}$ is a decreasing function on $(0,1)$, we have $(1-\beta+\frac{1}{N})\ln \frac{e}{1-\beta+\frac{1}{N}}\ge
(1-\gamma)\ln\frac{e}{1-\gamma}$ for all $m\ge N$.
Using the same notation as in Lemma~\ref{lem:small:m}, it follows from \eqref{est:mucst:upper} in Corollary~\ref{cor:alpha:mucst} that
\begin{equation}\label{ucst:alpha:min}
\begin{split}
\ucst_{\beta}(\alpha):=&\sup_{m\in \N}  \ucst_{\beta,m}(\alpha)\\ \le &\max\left(\ucst_{\beta,1}(\alpha),\ldots,
\ucst_{\beta,N-1}(\alpha),\left(\sqrt{2\left(1-\beta+\frac{1}{N}\right)\ln \frac{e}{1-\beta+\frac{1}{N}}}+\sqrt{2\alpha}\right)^2\right).
\end{split}
\end{equation}
Taking $\alpha\to 0^+$ in the above inequality, we deduce from Lemma~\ref{lem:small:m} that
\begin{align*}
\ucst_\beta
:=&\lim_{\alpha\to 0^+} \ucst_{\beta}(\alpha)\le \lim_{\alpha\to 0^+}
\max\left(\ucst_{\beta,1}(\alpha),\ldots,
\ucst_{\beta,N-1}(\alpha),\left(\sqrt{2(1-\beta+\tfrac{1}{N})\ln \frac{e}{1-\beta+\frac{1}{N}}}+\sqrt{2\alpha}\right)^2\right)\\
=&2\left(1-\beta+\tfrac{1}{N}\right)\ln \frac{e}{1-\beta+\tfrac{1}{N}}.
\end{align*}
Since $N>\frac{1}{\beta}$ can be arbitrarily large, combining with \eqref{est:mucst}, we proved
\[
(1-\beta)\max\left(c_g^2\ln \frac{2}{1-\beta},\frac{\pi}{2} \beta^2\right)\le
%\mathring{\ucst}^{\min}_{\beta,\infty} \le
\mathring{\ucst}^{\min}_{\beta}\le \ucst_\beta \le 2(1-\beta)\ln \frac{e}{1-\beta}.
\]
This proves the left-hand side of \eqref{eq:omega:2}.
\end{proof}

\section{Proofs of Corollaries}

We now provide proofs to Corollaries~\ref{cor:JL} and \ref{cor:rip} as well as the proofs of Corollaries~\ref{cor:JL1} and \ref{cor:rip1}.
%using the well-known techniques of union bounds in the literature.

\begin{proof}[Proof of Corollary~\ref{cor:JL}]
We assume that all the points $p_1,\ldots,p_N$ are distinct.  For every $T\in T_{\epsilon,\alpha}$,
we have $\abs{T^c}\leq \beta m$ for all $0\leq \beta\leq \frac{1-\sqrt{\alpha}}{32}\frac{\epsilon}{\ln \frac{1}{\epsilon}}$. By Theorem~\ref{thm:lowerbound:0},
for $j\ne k$, we have
\[
\PP\left\{ \left| \frac{ \|A_T p_j-A_T p_k\|^2}{m\|p_j-p_k\|^2}-1\right|\le \epsilon\; \mbox{for all}\, T\in T_{\epsilon,\alpha}\right\}\ge 1-3 e^{-\alpha(\epsilon^2/4-\epsilon^3/6)m},\qquad \forall\, m\in \N.
\]
Here we used the inequality
\[
0<\beta \leq \frac{1-\sqrt{\alpha}}{32}\frac{\epsilon}{\ln \frac{1}{\epsilon}}\leq \frac{(1-\sqrt{\alpha})\epsilon}{16 \ln \frac{4}{(1-\sqrt{\alpha})\epsilon}},
\]
provided that $0<\epsilon<\frac{1-\sqrt{\alpha}}{4}$.
Since there are $\binom{N}{2}=\frac{N(N-1)}{2}$ pairs $\{p_j,p_k\}$ with $j\ne k$, $j,k=1,\ldots,N$, using union bounds, we conclude that
\[
\PP\left\{ \left| \frac{ \|A_T p_j-A_T p_k\|^2}{m\|p_j-p_k\|^2}-1\right|\le \epsilon\; \forall\, T\in T_{\epsilon,\beta}, j\ne k, j,k=1,\ldots,N\right\}\ge 1-\frac{3N(N-1)}{2}e^{-\alpha(\epsilon^2/4-\epsilon^3/6)m}>0,
\]
where we used the assumption of $m>\frac{\ln(3N(N-1)/2)}{\alpha (\epsilon^2/4-\epsilon^3/6)}$ in the last inequality. This proves that \eqref{JLlemma:robust} holds with probability at least $1-\frac{3N(N-1)}{2}e^{-\alpha(\epsilon^2/4-\epsilon^3/6)m}>0$.
\end{proof}

\begin{proof}[Proof of Corollary~\ref{cor:rip}]
We slightly modify the argument in \cite[Lemma~5.1]{BDDW:ca:2008}. Let $\Lambda \subseteq \{1,\ldots,n\}$ with $|\Lambda|=s$. Set
$
\R^\Lambda:=\{x\in \R^n \, : \, x\; \mbox{is supported inside}\; \Lambda\}
$
and
$S^\Lambda:=\{ x\in \R^\Lambda \; : \; \|x\|=1\}$.
It is well known that there exists a subset $Q_{\Lambda,\epsilon}\subset S^\Lambda$ such that $|Q_{\Lambda,\epsilon}|\le (24/\epsilon)^s$ and $S^\Lambda\subseteq \cup_{\zeta\in Q_{\Lambda,\epsilon}} \{ x\in \R^n \; : \; \|x-\zeta\|\le \epsilon/8\}$. By Theorem~\ref{thm:lowerbound:0}, with probability at least
$1-3 (24/\epsilon)^s e^{-\alpha(\epsilon^2/16-\epsilon^3/24)m}$, we have
\begin{equation}\label{eq:high}
\sqrt{1-\epsilon/2}\|v\| \le \frac{1}{\sqrt{m}}\|A_T v\|\le \sqrt{1+\epsilon/2},\qquad \forall\; T\in T_{\epsilon/2,\alpha}\;\; \mbox{and}\;\; v\in Q_{\Lambda,\epsilon}.
\end{equation}
We next consider the case where  $A$ satisfies (\ref{eq:high}).
Define $\lambda:=\sup\{\frac{1}{\sqrt{m}}\|A_Tx\| \; : \; x\in S^\Lambda, T\in T_{\epsilon/2,\alpha}\}$. For every $x\in S^\Lambda$, there exists $v_x\in Q_{\Lambda,\epsilon}$ such that $\|x-v_x\|\le \epsilon/8$ and hence,
\[
\frac{1}{\sqrt{m}}\|A_T x\|\le \frac{1}{\sqrt{m}}\|A_T v_x\|+\frac{1}{\sqrt{m}}\|A(x-v_x)\|\le \sqrt{1+\epsilon/2}+\lambda \|x-v_x\|\le \sqrt{1+\epsilon/2}+\lambda \epsilon/8.
\]
By the definition of $\lambda$, we must have $\lambda\le \sqrt{1+\epsilon/2}+\lambda\epsilon/8$, which implies that
\[
\lambda \le \tfrac{\sqrt{1+\epsilon/2}}{1-\epsilon/8}\le \sqrt{1+\epsilon}
\]
 for all $0<\epsilon<1$. Therefore, for all $x\in \R^\Lambda$ and $T\in T_{\epsilon/2,\alpha}$,
$\frac{1}{\sqrt{m}}\|A_T x\|\le \lambda \|x\|\le \sqrt{1+\epsilon}\|x\|$ and
\[
\frac{1}{\sqrt{m}}\|A_T x\|\ge \frac{1}{\sqrt{m}} \|A_T v_x\|-\frac{1}{\sqrt{m}}\|A_T(x-v_x)\|\ge \sqrt{1-\tfrac{\epsilon}{2}}-\lambda\tfrac{\epsilon}{8} \ge \sqrt{1-\tfrac{\epsilon}{2}}-\tfrac{\epsilon}{8}
\sqrt{1+\epsilon}
\ge \sqrt{1-\epsilon},
\]
where the last inequality holds for all $0\le \epsilon\le 1$.
Thus, with probability at least
$1-3 (24/\epsilon)^s e^{-\alpha(\epsilon^2/16-\epsilon^3/24)m}$,
\begin{equation}\label{rip:est:0}
(1-\epsilon)\|x\|^2\le \tfrac{1}{m} \|A_T x\|^2\le (1+\epsilon)\|x\|^2, \qquad \forall\, x\in \R^\Lambda, T\in T_{\epsilon/2,\alpha}.
\end{equation}
Note that there are total $\binom{n}{s}\le (en/s)^s$ such subsets $\Lambda$. Therefore, \eqref{rip:est:0} holds for every such subset $\Lambda$. By union bounds, \eqref{rip:robust} holds with probability at least $1-3 (\tfrac{24en}{\epsilon s})^s e^{-\alpha(\epsilon^2/16-\epsilon^3/24)m}>0$ by our assumption
$s \ln \frac{24 en}{\epsilon s}< \alpha(\epsilon^2/16-\epsilon^3/24)m-\ln3$.
\end{proof}

\begin{proof}[Proof of Corollary~\ref{cor:JL1}]
The condition $0<\alpha<\frac{\pi}{12}(1-\beta)^2 h_\beta$
guarantees that $0<\lcst<\infty$, while the condition $m\ge \frac{1}{1-\beta}$ guarantees $0<\ucst<\infty$ (if $m=\frac{1}{1-\beta}$, then $0\ln\frac{e}{0}$ is understood as $\lim_{x\to 0^+} x\ln \frac{e}{x}=0$.)

By the left-hand inequality in \eqref{est:mlcst} of Corollary~\ref{cor:alpha:mlcst}, for any $x_0\in \R^n$ with $\|x_0\|=1$, we have $\PP(\I_{[\lcst,\infty],\beta})\ge 1-e^{-\alpha m}$. By \eqref{est:mucst:0} in Corollary~\ref{cor:alpha:mucst} and $\beta-\frac{1}{m}\le \gamma\le \beta$, noting that $(1-x)\ln \frac{e}{1-x}$ is a decreasing function on $(0,1)$, we deduce that $\PP(\I_{[0,\ucst],\beta})\ge 1-e^{-\alpha m}$.
Consequently, we have
\[
\PP\left\{ \lcst \|p_j-p_k\|^2\le \tfrac{1}{m}\|A_T p_j-A_T p_k\|^2 \le \ucst \|p_j-p_k\|^2, \; \forall\, T\subseteq \{1,\ldots, m\}, |T^c|\le \beta m \right\} \ge 1-2 e^{-\alpha m},\quad \forall\, m\in \N
\]
for every $j,k=1,\ldots, N$.
Since there are $\binom{N}{2}=\frac{N(N-1)}{2}$ pairs $\{p_j,p_j\}$ with $j\ne k$, $j,k=1,\ldots,N$, we conclude that
\eqref{JL:beta} holds with probability at least $1-N(N-1) e^{-\alpha m}>0$ for all $m\in \N$ by $m>\frac{1}{\alpha}\ln \frac{1}{N(N-1)}$.
\end{proof}

\begin{proof}[Proof of Corollary~\ref{cor:rip1}]
We use the same notation as in the proof of Corollary~\ref{cor:rip}. Define $T_{\le \beta}:=\{ T\subseteq \{1,\ldots,m\}\; : \; |T^c|\le \beta m\}$.
By Corollaries~\ref{cor:alpha:mlcst} and~\ref{cor:alpha:mucst}, with probability at least $1-2e^{-\alpha m}$,
\begin{equation}\label{eq:masat}
\sqrt{\lcst(1-\epsilon/2)}\|v\| \le \frac{1}{\sqrt{m}}\|A_T v\|\le \sqrt{\ucst(1+\epsilon/2)},\qquad \forall\; T\in T_{\le \beta}.
\end{equation}
We next consider the case where $A$ satisfies (\ref{eq:masat}).
Define $\lambda:=\sup\{\frac{1}{\sqrt{m}}\|A_Tx\| \; : \; x\in S^\Lambda, T\in T_{\le \beta}\}$. For every $x\in S^\Lambda$, there exists $v_x\in Q_{\Lambda,\epsilon}$ such that $\|x-v_x\|\le \epsilon/8$ and hence,
\[
\frac{1}{\sqrt{m}}\|A_T x\|\le \frac{1}{\sqrt{m}}\|A_T v_x\|+\frac{1}{\sqrt{m}}\|A(x-v_x)\|\le \sqrt{\ucst(1+\epsilon/2)}+\lambda \|x-v_x\|\le \sqrt{\ucst(1+\epsilon/2)}+\lambda \epsilon/8.
\]
By the definition of $\lambda$, we must have $\lambda\le \sqrt{\ucst(1+\epsilon/2)}+\lambda\epsilon/8$, from which we have $\lambda \le \sqrt{\ucst(1+\epsilon/2)}/(1-\epsilon/8)\le \sqrt{\ucst(1+\epsilon)}$ for all $0<\epsilon<1$. Therefore, for all $x\in \R^\Lambda$ and $T\in T_{\le \ucst}$,
$\frac{1}{\sqrt{m}}\|A_T x\|\le \lambda \|x\|\le \sqrt{\ucst(1+\epsilon)}\|x\|$ and
% by $0<\lcst\le \ucst$,
%
\begin{align*}
\frac{1}{\sqrt{m}}\|A_T x\|
&\ge \frac{1}{\sqrt{m}} \|A_T v_x\|-\frac{1}{\sqrt{m}}\|A_T(x-v_x)\|\ge \sqrt{\lcst} \sqrt{1-\epsilon/2}-\lambda\epsilon/8\\
&\ge \sqrt{\lcst}\left(\sqrt{1-\epsilon/2}-
\sqrt{\ucst/\lcst}  \sqrt{1+\epsilon} (\epsilon/8) \right)
\ge \sqrt{\lcst}\left(\sqrt{1-\epsilon/2}-
\sqrt{1+\epsilon} (\epsilon/8) \right)
\ge \sqrt{\lcst}\sqrt{1-\epsilon},
\end{align*}
where we used the fact that $\ucst/\lcst\ge 1$ by $0<\lcst\le \ucst$.
Thus, with probability at least
$1-2 (24/\epsilon)^s e^{-\alpha m}$,
\begin{equation}\label{rip:est:2}
\lcst(1-\epsilon)\|x\|^2\le \tfrac{1}{m} \|A_T x\|^2\le \ucst(1+\epsilon)\|x\|^2, \qquad \forall\, x\in \R^\Lambda, T\in T_{\le \beta}.
\end{equation}
Note that there are total $\binom{n}{s}\le (en/s)^s$ such subsets $\Lambda$. Therefore, \eqref{rip:est:2} holds for every such subset $\Lambda$. Hence, \eqref{rip:robust} holds with probability at least $1-2 (\tfrac{24en}{\epsilon s})^s e^{-\alpha m}>0$ by $s\ln \frac{24 en}{\epsilon s}<\alpha m-\ln 2$.
\end{proof}

\end{document}